\documentclass[11pt,a4paper,twoside]{article}

\usepackage[T1]{fontenc} 
\usepackage[utf8]{inputenc} 
\usepackage[english]{babel} 
\usepackage{amsmath,amsthm,amssymb,amscd,mathrsfs,amsfonts,epic,latexsym,hyperref,natbib,hypernat}
\usepackage[dvips]{graphicx}
\usepackage{psfrag,epsfig}
\usepackage[usenames]{color}
\usepackage[ruled,vlined]{algorithm2e}

\newtheorem{thm}{Theorem}
\newtheorem{defi}{Definition}
\newtheorem{prop}{Proposition}
\newtheorem{cor}{Corollary}

\newtheorem{lem}{Lemma}

\newtheorem{rem}{Remark}

\newcommand{\pr}{\mathbb{P}}

\newcommand{\esp}{\mathbb{E}}

\newcommand{\Var}{\mathbb{V}\text{ar}}
\newcommand{\Cov}{\mathbb{C}\text{ov}}

\newcommand{\R}{\mathbb{R}}

\newcommand{\FDR}{\mbox{FDR}}
\newcommand{\pFDR}{\mbox{pFDR}}
\newcommand{\lfdr}{\mbox{$\ell$FDR}}

\title{Nonparametric estimation of the density of the alternative
  hypothesis in a multiple testing setup. Application to local false
  discovery rate estimation}
\author{Van Hanh Nguyen$^{1,2}$ and Catherine Matias$^2$}

\begin{document}
\thispagestyle{empty}
\DeclareGraphicsExtensions{.pdf, .jpg, .jpeg, .png, .eps, .ps}

\maketitle

\begin{center}
1.  Laboratoire  de Mathématiques  d'Orsay, Université Paris  Sud, UMR
CNRS 8628, Bâtiment 425, 91~405 Orsay Cedex, France. 
E-mail: nvanhanh@genopole.cnrs.fr \\
2. Laboratoire Statistique et Génome, Université d'Évry Val d'Essonne, 
UMR CNRS 8071- USC INRA, 23 bvd de France, 91~037 Évry, France.
E-mail:  catherine.matias@genopole.cnrs.fr
\end{center}

\begin{abstract}
In a multiple testing context, we consider a semiparametric mixture model with
two components  where one  component is known  and corresponds  to the
distribution of $p$-values under the null
hypothesis and the other component $f$ is nonparametric and stands for
the distribution  under the  alternative hypothesis. Motivated  by the
issue of local false discovery rate estimation, we focus here on the
estimation   of   the   nonparametric   unknown   component $f$  in   the
mixture, relying on  a preliminary estimator of  the unknown proportion
$\theta$ of true null hypotheses. We propose and study the asymptotic properties of two 
different estimators  for this unknown component.  The first estimator
is a randomly weighted kernel estimator.  We establish an upper bound for its pointwise quadratic
risk,  exhibiting the classical nonparametric rate of convergence over a  class
of Hölder densities. To our knowledge, this is the first result
establishing convergence as well as corresponding rate for the
estimation of the unknown component in this nonparametric mixture. 
The second estimator is a maximum smoothed likelihood estimator. It is
computed through an iterative algorithm, for which we establish a
descent property.
In addition, these estimators are used in a multiple testing
procedure in order to estimate the local false discovery rate. Their 
respective performances are then compared on synthetic data.
\end{abstract}

\smallskip
\noindent {\it Key words and phrases: False
  discovery rate; kernel estimation; local false discovery rate; maximum smoothed likelihood; multiple testing;
  $p$-values; semiparametric mixture model.}\\

%\noindent \MSC[2008] 62G20, 62G10
% 62G07 Density estimation
% 62G10 Hypothesis testing 
% 62G20 Asymptotic properties

%%%%%%%%%%%%%%%%

\section{Introduction}            

In the framework of multiple testing problems (microarray analysis,
neuro-imaging, etc), a mixture model with two populations is considered
\begin{equation}
\label{eq-model}
\forall x \in \R^d, \quad 
g(x) = \theta \phi(x)+(1-\theta)f(x),
\end{equation}
where $\theta$ is the unknown proportion of true null hypotheses,
$\phi$ and $f$ are the densities of the observations generated
under the null and alternative hypotheses, respectively. More
precisely, assume the test statistics are independent and identically
distributed (iid) with a continuous distribution under the
corresponding null hypotheses and we observe
the $p$-values $X_1, X_2,\ldots, X_n$ associated with $n$ independent
tested hypotheses, then the density function $\phi$ is the uniform
distribution on $[0,1]$ while the density function $f$ is assumed unknown. The parameters of the model are $(\theta,f)$, where
  $\theta$   is    a   Euclidean    parameter   while   $f$    is   an
  infinite-dimensional one and the model becomes
\begin{equation}
\label{eq-model-pvalue}
\forall x \in [0,1], \quad 
g(x) = \theta+(1-\theta)f(x) . 
\end{equation}
In the following, we focus on model~\eqref{eq-model-pvalue} that is slightly simpler than~\eqref{eq-model}. A central
problem in the multiple testing setup is the control of type I
(\emph{i.e.} false positive) and type II (\emph{i.e.} false negative) errors. The
most popular criterion regarding type I errors is the false discovery
rate (\FDR), proposed by \cite{Benjamini1995}.  To set up
the notation, let $H_i$ be the $i$-th (null) hypothesis. The outcome
of testing $n$ hypotheses simultaneously can be summarized as indicated in Table~\ref{tab:outcomes}.

\begin{table}[h!]
  \centering
\caption{Possible outcomes from testing $n$ hypotheses $H_1,\dots,H_n$.}
\label{tab:outcomes}
  \begin{tabular}{llll}
  \hline
            & Accepts $H_i$ & Rejects $H_i$ & Total \\
  \hline
  $H_i$ is true & TN & FP & $n_0$ \\
  $H_i$ is false & FN & TP & $n_1$ \\
Total & N & P & $n$\\
  \hline
\end{tabular}
\end{table}

 \cite{Benjamini1995} define \FDR\ as the expected
 proportion of rejections that are incorrect,
\begin{equation*}
\FDR = \esp \Big[\frac{\mbox{FP}}{\max(\mbox{P}, 1)}\Big] = \esp \Big[\frac{\mbox{FP}}{\mbox{P}}
\big| \mbox{P} > 0\Big] \mathbb{P}(\mbox{P} >0).
\end{equation*}
They provide a multiple testing procedure that guarantees the bound $\FDR \leq
\alpha$, for  a desired level $\alpha$.  \cite{Storey2003} proposes to
modify \FDR\ so as to obtain a new criterion, the positive \FDR\ (or \pFDR), defined by 
\begin{equation*}
\pFDR = \esp \Big[\frac{\mbox{FP}}{\mbox{P}}
\big| \mbox{P} > 0\Big],
\end{equation*}
and argues that it is conceptually more sound than \FDR. For microarray
data for instance, there is a large value of the number of hypotheses
$n$ and the difference between \pFDR\ and \FDR\ is
generally small as the extra factor $\mathbb{P}(P >0)$ is very close
to $1$ \citep[see][]{Liao2004}. In  a mixture context,
the \pFDR\ is given by
\begin{equation*}
\pFDR(x) = \mathbb{P}( H_i\ \text{being true}\ | X \leq x) =
\frac{\theta \Phi(x)}{\theta \Phi(x)+(1-\theta)F(x)},
\end{equation*}
where $\Phi$ and $F$ are the cumulative distribution functions (cdfs) for densities
$\phi$ and $f$, respectively. (It is notationally convenient to
consider events of the form $X \leq x$, but we could just as well consider tail
areas to the right, two-tailed events, etc). 

\cite{Efron2001} define the local false discovery rate (\lfdr) to quantify
the plausibility of a particular hypothesis being true, given its
specific test statistic or $p$-value. In a mixture framework, the
\lfdr\ is the Bayes posterior probability
\begin{equation}
\label{eq-lfdr}
\lfdr(x) = \mathbb{P}( H_i\ \text{being true}\ | X=x) =
1-\frac{(1-\theta)f(x)}{\theta \phi(x)+(1-\theta)f(x)}.
\end{equation}
 In many multiple testing
frameworks, we need information at the individual level about the
probability for a given observation to be a false positive \citep{Aubert2004}. This motivates estimating the local false discovery rate
\lfdr. Moreover, the quantities \pFDR\ and \lfdr\ are analytically related by $\pFDR(x) = \esp [\lfdr(X)|  X \leq  x]$. As  a consequence  (and recalling  that  the difference
between \pFDR\ and \FDR\ is generally small), \cite{Robin2007} propose to estimate
\FDR\ by 
\begin{equation*}
\widehat{\FDR}(x_i) = \frac{1}{i} \sum_{j=1}^i \widehat{\lfdr}(x_j),
\end{equation*}
where $\widehat{\lfdr}$ is an estimator of \lfdr\ and the observations
$\{x_i\}$ are increasingly ordered.
A natural strategy to estimate \lfdr\ is to start by estimating both the
proportion  $\theta$ and  either $f$  or $g$.  Another  motivation for
estimating the parameters in this mixture model comes from the
works  of Sun  and Cai  (\citeyear{Sun_Cai07,Sun_Cai09}),  who develop
adaptive compound decision rules for false discovery rate control.  These 
 rules   are  based  on the estimation  of   the  parameters  in
model~\eqref{eq-model}   (dealing   with   $z$-scores)   rather   than
model~\eqref{eq-model-pvalue} (dealing with $p$-values).
However, it appears that in some very specific cases (when the alternative is symmetric about the null), the oracle
version of their  procedure based on the $p$-values  (and thus relying
on estimators of the parameters in model~\eqref{eq-model-pvalue}) may
outperform  the  one  based on  model~\eqref{eq-model}  \cite[see][for
more details]{Sun_Cai07}.  In the following, we are thus interested in
estimating parameters in model~\eqref{eq-model-pvalue}.\\

In a previous work \citep{Prop_True_Null}, we discussed the estimation
of    the   Euclidean    part   of    the   parameter    $\theta$   in
model~\eqref{eq-model-pvalue}. Thus, we will not consider further this point
here. We rather focus on the estimation of the unknown density $f$, relying on a preliminary estimator of $\theta$. We just mention that many estimators of $\theta$ have been proposed  in   the  literature. 
One of the most well-known is the one proposed by \cite{Storey2002}, motivating its use  in our simulations.  Some  of   these estimators   are proved to be consistent (under  suitable  model  assumptions). 
%This is for instance the case for  the  one proposed by \cite{Celisse-Robin2010}.  Moreover, this latter estimator has     been     shown      to     be     $\sqrt{n}$-consistent     in \cite{Prop_True_Null}. This will be used later when assessing the rate of convergence of one of our estimators of $f$.  
Of course, we will need some specific properties of estimators $\hat \theta_n$ of $\theta$ to obtain rates of convergence of estimators of $f$. Besides,  existence of estimators $\hat \theta_n$ satisfying those specific properties is a consequence of~\cite{Prop_True_Null}.

Now, different modeling assumptions on the marginal density $f$ have been
proposed in the literature.  For instance, parametric models have been
used  with  Beta  distribution   for  the  $p$-values  \citep[see  for
example][]{Allison2002, Pounds2003, Liao2004} or Gaussian distribution
of the probit transformation of the $p$-values \citep{McLachlan2006}. In the
framework of nonparametric estimation, \cite{Strimmer2008} proposed a
modified Grenander density estimator for $f$, which has been initially
suggested by
\cite{Langaas2005}. This approach requires monotonicity constraints on
the density $f$.  Other nonparametric approaches consist in
relying on regularity assumptions on $f$. This is done for instance in
\cite{Neuvial2010},  who  is primarily  interested  in estimating  $\theta$
under the assumption  that it is equal to $g(1)$.  Relying on a kernel
estimator of  $g$, he derives  nonparametric rates of  convergence for
$\theta$. 
Another kernel estimator  has been proposed by~\cite{Robin2007}, along
with  a  multiple  testing  procedure,  called  \texttt{kerfdr}.  This
iterative  algorithm   is  inspired  by   an  expectation-maximization
(\texttt{em}) procedure \citep{DempsterLR}. It 
is  proved  to  be convergent  as  the number  of
iterations increases.  However, it does not optimize any criterion and
contrarily to the original \texttt{em} algorithm, it does not increase the
observed data likelihood function. Besides, the
asymptotic properties (with the number of hypotheses $n$) of the kernel estimator underlying \citeauthor{Robin2007}'s approach have not
been studied. Indeed, its iterative form prevents from obtaining any theoretical
result on its convergence properties.

The first part of the present work focuses on the properties of a randomly
weighted  kernel estimator,  which in  essence, is  very similar  to the
iterative approach proposed by \cite{Robin2007}.  Thus, this part may
be viewed as a theoretical validation of \texttt{kerfdr}  approach that  gives some
insights  about  the  convergence   properties  (as  the  sample  size
increases) of this method. In particular, we establish that relying on
a preliminary estimator of $\theta$ that roughly converges at parametric rate (see exact condition in Corollary~\ref{cor1}),
we obtain  an estimator of the  unknown density $f$  that converges at
the usual minimax nonparametric rate. To our knowledge, this is the first result
establishing convergence as well as corresponding rate for the
estimation of the unknown component in model~\eqref{eq-model-pvalue}. 
In a  second part, we are  interested in a
new iterative  algorithm for estimating the unknown  density $f$, that
aims at maximizing a smoothed likelihood.  We refer to Paragraph 4.1 in \cite{Egg_LaR_book} for an
interesting  presentation  of kernel  estimators  as maximum  smoothed
likelihood  ones.   Here,  we  base   our  approach  on  the  work  of
\cite{Levine}, who  study a  maximum smoothed likelihood  estimator for
multivariate mixtures.  The main idea consists in introducing a nonlinear
smoothing operator on the unknown component $f$ as proposed in \cite{Egg_LaR_95}. 
We prove that the resulting algorithm possesses a desirable descent property, just as
an  \texttt{em}  algorithm does.  We  also  show  that it is competitive with respect to  \texttt{kerfdr} algorithm, both when used to estimate $f$ or \lfdr.

The  article  is organized  as  follows. In Section~\ref{sec:algorithms}, we start by describing different procedures to estimate $f$. We distinguish two types of procedures and first describe direct (non iterative) ones in Section~\ref{sec:direct}. We mention a direct naive approach but the main procedure from this section is a randomly weighted kernel estimator. Then, we switch to iterative procedures (Section~\ref{sec:iterative}). The first one is not new: \texttt{kerfdr} has been proposed in \cite{Robin2007,Guedj}. The second  one, called \texttt{msl}, is new and  adapted from the work of \cite{Levine} in a different context (multivariate mixtures). These iterative procedures are  expected to be more accurate than direct ones, but their properties are in general more difficult to establish. As such, the direct randomly weighted kernel estimator from Section~\ref{sec:direct} may be viewed as a proxy for studying the convergence properties (with respect to $f$) of \texttt{kerfdr} procedure (properties that are unknown). Section~\ref{sec:results} then gives the theoretical properties of the procedures described in  Section~\ref{sec:algorithms}. In particular, we establish (Theorem~\ref{thm:cv_f}) an upper bound on the pointwise quadratic risk of the randomly weighted kernel procedure. Moreover, we prove that \texttt{msl} procedure  possesses a descent property with respect to some criterion (Proposition~\ref{prop:hatft_decreases}).
In Section~\ref{sec:lfdr}, we rely on our different estimators to estimate both density $f$ and the local  false discovery rate \lfdr. We  present simulated experiments to compare  their performances.  All the  proofs have  been  postponed to
Section~\ref{sec:proofs}. Moreover, some  of the more technical proofs
have been further postponed to Appendix~\ref{sec:appendix}.

\section{Algorithmic procedures to estimate the density $f$}\label{sec:algorithms}

\subsection{Direct procedures}\label{sec:direct}
Let us be  given a  preliminary estimator $\hat{\theta}_n$ of $\theta$ as well as a  nonparametric estimator $\hat{g}_n$ of $g$.  We propose  here to rely on  a 
kernel estimator of the density $g$
\begin{equation}\label{eq:vallee_poussin}
\hat{g}_n(x) = \frac{1}{nh} \sum_{i=1}^n K\Big(\frac{x-X_i}{h}\Big) = \frac 1 n \sum_{i=1}^n K_{i,h}(x),
\end{equation} 
where $K$ is a  kernel (namely a real-valued integrable
function such that $\int K(u)du$ $ = 1$),  $h>0$ is a bandwidth (both are to be chosen later) and
\begin{equation}
  \label{eq:Kih}
K_{i,h}(\cdot)=\frac 1 h K\Big(\frac{\cdot-X_i}{h}\Big).
\end{equation}
Note that this estimator of $g$ is consistent under appropriate assumptions.

\paragraph*{A naive approach.}
 From Equation~\eqref{eq-model-pvalue}, it is natural to propose to estimate  $f$ with 
\begin{equation*}
\hat{f}_n^{\text{naive}}(x) = \frac{\hat{g}_n(x)-\hat{\theta}_n}{1 - \hat{\theta}_n}\mathbf{1}_{\{\hat{\theta}_n \neq 1 \}},
\end{equation*}  
where $1_A$ is the indicator function of set $A$.   
This estimator has the same theoretical properties as the randomly weighted kernel estimator presented below. However,  it is much worse in practice, as we shall see in the simulations of Section~\ref{sec:lfdr}.

\paragraph*{A randomly weighted kernel estimator.} 
We now explain a natural construction for an estimator of $f$ relying on a randomly weighted version of a kernel estimator of $g$.
 For any hypothesis, we introduce  a (latent) random variable $Z_i$ that equals $0$ if the null hypothesis $H_i$ is true and $1$ otherwise, 
\begin{equation}
  \label{eq:Z}
  \forall i=1,\dots, n \quad Z_i=\left\{
    \begin{array}[]{ll}
 0 & \mbox{if } H_i \mbox{ is true,}    \\ 
1 & \mbox{otherwise.}     
    \end{array}
\right.
\end{equation}
Intuitively, it would be convenient to introduce a weight for each observation $X_i$, meant to select this observation only if it comes from $f$. 
Equivalently, the weights are used to select the indexes $i$ such that $Z_i=1$. Thus, a natural kernel estimate of $f$ would be  
\begin{equation*}
f_1(x)     =    \frac{1}{h}\sum_{i=1}^n\frac{Z_i}{\sum_{k=1}^n    Z_k}
K\Big(\frac{x-X_i}{h}\Big)
= \sum_{i=1}^n\frac{Z_i}{\sum_{k=1}^n    Z_k} K_{i,h}(x),
\ x \in [0,1]. 
\end{equation*}
However, $f_1$ is not an estimator and cannot be directly used since the random variables $Z_i$ are
not observed. A natural approach \citep[initially proposed in][]{Robin2007} is to replace them with their conditional expectation given the data $\{X_i\}_{1\le i \le n}$, namely with the posterior probabilities
$\tau (X_i) =\esp (Z_i|X_i)$  defined by
\begin{equation}
\label{eq:posterior-proba}
\forall x \in [0,1], \ \tau(x) = \esp (Z_i|X_i=x)=\frac{(1-\theta)f(x)}{g(x)}=1 - \frac{\theta}{g(x)}.
\end{equation}
This leads to  the following definition  
\begin{equation}\label{eq:f2}
\forall x \in [0,1], \ f_2(x) = \sum_{i=1}^n \frac{\tau(X_i)} {\sum_{k=1}^n \tau(X_k)} K_{i,h}(x).
\end{equation}
Once again, the weight $\tau_i =\tau(X_i)$ depends on the unknown parameters
$\theta$ and  $f$ and  thus $f_2$  is not an  estimator but  rather an
oracle. To solve this  problem, \cite{Robin2007} proposed an iterative
approach,   called    \texttt{kerfdr}   and   discussed    below,   to
approximate~\eqref{eq:f2}.  For the moment, we propose to replace the posterior
probabilities $\tau_i$ by direct (rather than iterative) estimators to obtain a randomly weighted kernel
estimator of $f$.  
%%%
%%%
Specifically, we propose to estimate  the posterior
probability $\tau(x)$  by
\begin{equation}\label{eq:hat_tau}
\forall x \in [0,1], \ \hat{\tau}(x) = 1 - \frac{\hat{\theta}_n}{\hat{g}_n(x)}.
\end{equation}
Then, by defining the weight 
\begin{equation}\label{eq:def_weight}
\hat{\tau}_i = \hat{\tau}(X_i)
=1 - \frac{\hat{\theta}_n}{\tilde{g}_n(X_i)}, 
\ \text{where}\ \tilde{g}_n(X_i)=
\frac{1}{(n-1)} \sum_{j \neq i}^n K_{j,h}(X_i),
\end{equation}
we get a randomly  weighted kernel estimator of the density $f$ defined as 
\begin{equation}\label{eq:rwk}
\forall x \in [0,1], \ \hat{f}_n^{\text{rwk}}(x) = \sum_{i=1}^n \frac{\hat{\tau}_i}{\sum_{k=1}^n \hat{\tau}_k} K_{i,h}(x).
\end{equation}
Note that it is not necessary to use the same kernel $K$ in defining $\hat
g_n$ and $\hat f_n^{\text{rwk}}$, nor the same bandwidth $h$. In practice, we rely on the same kernel chosen with a compact support (to avoid boundary effects) and as we will see in Section~\ref{sec:results}, the bandwidths have to be chosen of the same order.   
%However in practise, the choice of the kernel with compact support has a good influence on the performances of our estimators. 
Also note that the slight modification from $\hat g_n$ to $\tilde g_n$ in defining the weights~\eqref{eq:def_weight} is minor and used in practice to reduce the bias of $\tilde g_n(X_i)$.

\subsection{Iterative procedures}\label{sec:iterative}
In   this  section,  we   still  rely   on  a   preliminary  estimator
$\hat{\theta}_n$ of $\theta$.  Two different procedures are described:
\texttt{kerfdr}  algorithm, proposed  by~\cite{Robin2007,Guedj}  and a
maximum smoothed likelihood \texttt{msl} estimator, inspired from the work of~\cite{Levine} in the context of multivariate nonparametric mixtures.
Both rely on an iterative randomly weighted kernel approach. 
The general form of these procedures is described by Algorithm~\ref{algo:iterative}. 
The main difference between the two procedures lies in the choice of the functions $\tilde K_{i,h}$ (that play the role of a kernel) and the way the weights are updated. 

\begin{algorithm}
  %\dontprintsemicolon
 \CommentSty{// Initialization}\;
Set initial weights $\hat{\omega}_i^0 \sim \mathcal{U}\big([0,1]\big), i=1,2,\ldots,n. $
 \BlankLine
 \BlankLine
  \While{$\max_i|\hat{\omega}_i^{(s)} -
\hat{\omega}_i^{(s-1)}|/\hat{\omega}_i^{(s-1)} \ge  \epsilon$}{
 \BlankLine
 \BlankLine
 \CommentSty{// Update estimation of $f$}\; 
$\hat{f}^{(s)}(x_i) =
\sum_j                                     \hat{\omega}_j^{(s-1)}\tilde
K_{j,h}(x_i)/\sum_k\hat{\omega}_k^{(s-1)}$\\
%, where $\tilde K_{j,h}$ is defined through~\eqref{}\\
 \BlankLine
 \BlankLine
\CommentSty{// Update of weights}\; 
 $\hat{\omega}_i^{(s)}$: depends on the procedure, see Equations~\eqref{eq:weight_kerfdr} and \eqref{eq:weight_msl}
% = (1 - \hat{\theta})\hat{f}^{(s)}(x_i)/\hat{g}^{(s)}(x_i)
 \BlankLine
 \BlankLine
 $s\leftarrow s+1$\;
}
 \BlankLine
 \BlankLine
\CommentSty{// Return}\;
$\hat{f}^{(s)}(\cdot) =
\sum_i\hat{\omega}_i^{(s-1)}\tilde K_{i,h}(\cdot)/\sum_k\hat{\omega}_k^{(s-1)}$ 
\caption{General structure of the iterative algorithms}
\label{algo:iterative}
\end{algorithm}

Note that the parameter $\theta$ is fixed throughout these iterative procedures. Indeed, as already noted by \cite{Robin2007}, the solution $\theta=0$ is a fixed point of a modified \texttt{kerfdr} algorithm where $\theta$ would be iteratively updated. This is also the case with the maximum smoothed likelihood procedure described below in the particular setup of model~\eqref{eq-model-pvalue}. This is why we keep $\theta$ fixed in both procedures. We now describe more explicitly the two procedures.

\paragraph*{\texttt{Kerfdr} algorithm.}
This procedure has been  proposed by \cite{Robin2007,Guedj}
  as  an  approximation  to  the  estimator  suggested
by~\eqref{eq:f2}.  
In  this procedure,  functions  $\tilde K_{i,h}$  more simply  denoted
$K_{i,h}$  are defined  through~\eqref{eq:Kih} where  $K$ is  a kernel
(namely $\int K(u)du=1$) and following~\eqref{eq:posterior-proba}, the weights are updated as follows
\begin{equation}
  \label{eq:weight_kerfdr}
\hat{\omega}_i^{(s)} = \frac{(1 - \hat{\theta}_n)\hat{f}^{(s)}(x_i)} {  \hat{\theta}_n + (1 - \hat{\theta}_n)\hat{f}^{(s)}(x_i)}.
\end{equation}

This       algorithm      has       some       \texttt{em}      flavor
\citep{DempsterLR}. Actually, updating the weights $\hat{\omega}_i^{(s)}$
is equivalent to \texttt{expectation}-step, and $\hat{f}^{(s)}(x)$ can
be  seen as  an average  of $\{K_{i,h}(x)\}_{1\le  i \le  n}$  so that
updating    the    estimator     $\hat{f}$    may    look    like    a
\texttt{maximization}-step. However, as noted in~\cite{Robin2007}, the algorithm
does not optimize any given criterion. Besides, it does not increase the
observed data likelihood function.

The relation between $\hat{f}^{(s)}$ and $\hat{\omega}^{(s)}$ implies that the sequence
$\{\hat{\omega}^{(s)}\}_{s\ge 0}$ satisfies $\hat{\omega}^{(s)} =
\psi(\hat{\omega}^{(s-1)})$, where
\begin{eqnarray*}
\psi : [0,1]^n\backslash\{0\}\rightarrow [0,1]^n,\quad
\psi_i(u)=\frac{\sum_iu_ib_{ij}}{\sum_iu_ib_{ij}+\sum_iu_i}, \quad
\text{with}\quad b_{ij}=\frac{1-\hat{\theta}_n}{\hat{\theta}_n}\times \frac{K_{i,h}(x_j)}{\phi(x_j)}.
\end{eqnarray*}

Thus, if the sequence $\{\hat{\omega}^{(s)}\}_{s\ge 0}$ is convergent, it has to
converge towards a fixed  point of $\psi$. \cite{Robin2007} prove that
under some mild conditions, \texttt{kerfdr} estimator  is self-consistent, meaning that as the number of iterations $s$ increases, the
sequence $\hat{f}^{(s)}$ converges towards the function 
\begin{equation*}
f_3(x) =  \sum_{i=1}^n\frac{ \hat{\omega}_i^*}{\sum_{k} \hat{\omega}_k^* } K_{i,h}(x), 
\end{equation*}
where $\hat{\omega}_i^*$ is the (unique) limit of $\{\hat{\omega}_i^{(s)}\}_{s\ge
  0}$.  Note that contrarily to $f_2$, function $f_3$ is a randomly weighted kernel estimator of
$f$. However, nothing is known about the convergence of $f_3$ nor $\hat{f}^{(s)}$
towards  the true  density $f$  when the sample size $n$  tends to  infinity  (while the
bandwidth $h=h_n$ tends to 0).   Indeed, the weights $\{\hat{\omega}_i^{(s)}\}_{s\ge
  0}$ used by the kernel estimator $\hat{f}^{(s)}$ form an iterative sequence.  Thus it is very difficult to study the convergence properties of this weight sequence or of the corresponding estimator.

 We thus  propose another  randomly weighted
kernel estimator, whose weights are slightly different from those used
in the construction of $\hat f^{(s)}$. More precisely, those weights are
not defined iteratively but they mimic the sequence of weights $\{\hat{\omega}_i^{(s)}\}_{s\ge   0}$. 

\paragraph*{Maximum smoothed likelihood estimator.}
Following  the  lines  of  \cite{Levine}, we  construct  an  iterative
estimator sequence of the density  $f$ that relies on the maximisation
of a smoothed likelihood. 
%This estimation procedure possesses a desirable monotone  property  of  a   smoothed  version  of  the  log-likelihood function, as we show below. 
Assume in the following that $K$ is a positive and symmetric kernel on $\mathbb{R}$. We define its rescaled version as
\[
K_h(x) = h^{-1}K(h^{-1}x). 
\]
We consider a linear
smoothing operator $\mathcal{S} : \mathbb{L}_1([0,1])
\rightarrow \mathbb{L}_1([0,1])$ defined as
 \begin{equation*}
\mathcal{S} f(x) = \int_0^1\frac{K_h(u-x)f(u)}{\int_0^1K_h(s-u)ds}du,\ \text{for all}\ x \in [0, 1].
\end{equation*}
We remark that if $f$ is a density on $[0,1]$ then $\mathcal{S} f$ is
also a density on $[0,1]$. Let us consider a submodel of
model~\eqref{eq-model-pvalue}    restricted    to   densities    $f\in
\mathcal{F}$ with 
 \begin{equation*}
\mathcal{F} = \{ \text{densities } f\ \text{on } [0,1]\ \text{such that }
\log f \in \mathbb{L}_1([0,1]) \}.
\end{equation*}
We denote by $\mathcal{S}^*: \mathbb{L}_1([0,1])
\rightarrow \mathbb{L}_1([0,1])$ the operator 
\begin{equation*}
\mathcal{S}^*f(x) = \frac{\int_0^1K_h(u-x)f(u)du}{\int_0^1K_h(s-x)ds}.
\end{equation*}
Note the difference between $\mathcal{S}$ and $\mathcal{S}^*$. The 
operator  $\mathcal{S}^*$   is  in   fact  the  adjoint   operator  of
$\mathcal{S}$. Here, we rely more  specifically on the earlier work of
\cite{Eggermont99} that takes into  account the case where the density
support ($[0,1]$  in our  case) is different  from the  kernel support
(usually $\mathbb{R}$).  Indeed in  this case, the normalisation terms
introduce a difference between $\mathcal{S}$ and $\mathcal{S}^*$.
Then for a density $f \in \mathcal{F}$, we approach it by a nonlinear
smoothing operator $\mathcal{N}$ defined as
 \begin{equation*}
\mathcal{N}f(x) = \exp\{(\mathcal{S}^*(\log f))(x)\}, \quad x\in [0,1].
\end{equation*}
Note that $\mathcal{N}f$ is not necessarily a density. 
Now, the maximum smoothed likelihood procedure consists in applying Algorithm~\ref{algo:iterative}, relying on 
\begin{equation}\label{eq:tildeKih}
\tilde K_{i,h} (x) =\frac {K_{i,h}(x)}{\int_0^1 K_{i,h}(s)ds},
\end{equation}
where  $K_{i,h}$  is   defined  through~\eqref{eq:Kih}  relying  on  a
positive symmetric kernel $K$ and 
\begin{equation}\label{eq:weight_msl}
\hat{\omega}_i^{(s)} = \frac{(1-\hat{\theta}_n)\mathcal{N}\hat{f}^{(s)}(x_i)}{\hat{\theta}_n+(1-\hat{\theta}_n)\mathcal{N}\hat{f}^{(s)}(x_i)}.
\end{equation}
In Section~\ref{sec:msl}, we explain where these choices come from and why this procedure corresponds to
a maximum smoothed likelihood approach. Let us remark that as in \texttt{kerfdr} algorithm, the sequence of weights $\{\hat{\omega}^{(s)}\}_{s\ge 0}$ also satisfies $\hat{\omega}^{(s)} = \varphi(\hat{\omega}^{(s-1)})$ for some specific function $\varphi$. Then, if the sequence $\{\hat{\omega}^{(s)}\}_{s\ge 0}$ is convergent, it must be convergent to a fixed point of $\varphi$. Existence and uniqueness of a fixed point for \texttt{msl} algorithm is explored below in Proposition~\ref{prop:subsequence}.
%In practice, when $\theta$ is iteratively updated, the sequence of weights $\{\hat{\omega}^{(s)}\}_{s\ge 0}$ will converge to $0$.

In the  following section, we  thus establish theoretical properties  of the
procedures presented here. These are then further compared on simulated data
in Section~\ref{sec:lfdr}.

\section{Mathematical properties of the algorithms}\label{sec:results}
\subsection{Randomly weighted kernel estimator}
We provide below the convergence properties of the estimator
$\hat{f}_n^{\text{rwk}}$ defined through~\eqref{eq:rwk}. In fact, these naturally depend on the properties of the
plug-in estimators $\hat{\theta}_n$ and $\hat{g}_n$. 
 We are interested  here in
 controlling the pointwise quadratic risk of $\hat{f}_n^{\text{rwk}}$. This
 is possible on a class of densities $f$ that are regular enough. In the
 following, 
%  we let  $\pr_g$ and  $\esp_g$ respectively  denote  probability and
% expectation  of iid random variables  with density $g$.  In the same
% way, 
we denote by  $\pr_{\theta,f}$   and  $\esp_{\theta,f}$  the   probability  and
 corresponding expectation in the more specific
 model~\eqref{eq-model-pvalue}. 
 Moreover, $\lfloor x \rfloor$ denotes the largest integer strictly
 smaller than $x$. Now, we recall that the order of a
  kernel is defined as its first nonzero moment \citep{Tsybakov-book} and we
  recall below the definition of Hölder classes of functions. 

\begin{defi}
Fix $\beta > 0, L > 0$ and denote by $H(\beta,L)$ the set of functions
$\psi :[0,1] \rightarrow \mathbb{R}$ that are $l$-times continuously
differentiable on $[0,1]$ with $ l = \lfloor \beta \rfloor$ and
satisfy
\begin{equation*}
| \psi^{(l)}(x) - \psi^{(l)}(y) | \leq L |x - y |^{\beta-l} ,\quad \forall
x,y \in [0,1].
\end{equation*}
The set $H(\beta,L)$ is called the $(\beta, L)$-Hölder class
of functions.
\end{defi}
We denote by $\Sigma(\beta,L)$ the set
\begin{equation*}
\Sigma(\beta,L) = \Big\{\psi : \psi\ \text{is a density on}\ [0, 1]\ \text{and}\ \psi \in H(\beta,L)\Big\}.
\end{equation*}
According to the proof of Theorem 1.1 in \cite{Tsybakov-book}, we
remark that  
\[
\sup_{\psi \in \Sigma(\beta,L)} \|\psi \|_{\infty} < +\infty.
\]
In order to obtain the rate of convergence of $\hat{f}_n^{\text{rwk}}$
to $f$, we introduce the following assumptions
%%%
\paragraph*{(A1)} \label{hyp:A1} The kernel $K$ is a right-continuous function.
\paragraph*{(A2)} \label{hyp:A2} $K$ is of bounded variation.
\paragraph*{(A3)} \label{hyp:A3} The kernel $K$ is of order $l=\lfloor \beta \rfloor$ and satisfies
\begin{equation*}
\int K(u)du=1,\ \int K^2(u)du < \infty,\ \text{and}\ \int
|u|^{\beta}|K(u)|du < \infty.
\end{equation*}
\paragraph*{(B1)} \label{hyp:B1} $f$ is a uniformly continuous density
function. 
\paragraph*{(C1)} \label{hyp:C1} The bandwidth $h$ is of order $\alpha n^{-1/(2\beta+1)}$, $\alpha
> 0$.\\

Note that there exist kernels satisfying
Assumptions~\hyperref[hyp:A1]{\textbf{(A1)}}-\hyperref[hyp:A3]{\textbf{(A3)}}
\cite[see for instance Section 1.2.2 in][]{Tsybakov-book}. Note also that if $f\in \Sigma(\beta,L)$, it automatically satisfies Assumption~\hyperref[hyp:B1]{\textbf{(B1)}}.

\begin{rem}
\label{rem:first}
\begin{itemize}
\item [i)]   We   first   remark   that   if    kernel   $K$   satisfies
Assumptions~\hyperref[hyp:A1]{\textbf{(A1)}},~\hyperref[hyp:A2]{\textbf{(A2)}} and if Assumptions~\hyperref[hyp:B1]{\textbf{(B1)}} and \hyperref[hyp:C1]{\textbf{(C1)}} hold, then the
kernel density estimator $\hat{g}_n$ defined by~\eqref{eq:vallee_poussin} converges uniformly almost surely to $g$ \citep{Wied}. In other words 
\[
\|\hat{g}_n - g\|_{\infty}
  \xrightarrow[n\to\infty] {a.s} 0.
\] 
\item[ii)] If kernel $K$ satisfies Assumption~\hyperref[hyp:A3]{\textbf{(A3)}} and if 
Assumption~\hyperref[hyp:C1]{\textbf{(C1)}} holds, then for all $n \geq 1$
\begin{equation*}
\sup_{x \in [0,1]}\sup_{f \in \Sigma(\beta,L)}
\esp_{\theta,f}(|\hat{g}_n(x)-g(x)|^{2}) \leq C
n^{\frac{-2\beta}{2\beta+1}},
\end{equation*}
where   $C=C(\beta,   L,   \alpha,   K)$  
\cite[see   Theorem   1.1   in][]{Tsybakov-book}.
\end{itemize}
\end{rem}

In  the following theorem,  we give the
rate  of  convergence to  zero  of  the  pointwise quadratic  risk  of
$\hat{f}_n^{\text{rwk}}$.
% defined by~\eqref{eq:rwk}.
\begin{thm}
\label{thm:cv_f}
Assume           that          kernel           $K$          satisfies
Assumptions~\hyperref[hyp:A1]{\textbf{(A1)}}\textbf{-}\hyperref[hyp:A3]{\textbf{(A3)}}
and $K \in\mathbb{L}_4(\mathbb{R})$. 
%and                kernel                $K$                satisfies
%Assumptions~\hyperref[hyp:A1]{\textbf{(A1)}}\textbf{-}\hyperref[hyp:A3]{\textbf{(A3)}}. 
If $\hat{\theta}_n$
converges almost surely to $\theta$ and the bandwidth $h=\alpha
n^{-1/(2\beta +  1)}$ with $\alpha > 0$, then for any $\delta >0$, the  pointwise quadratic risk  of $\hat{f}_n^{\text{rwk}}$
satisfies 
\begin{eqnarray*}
\sup_{x  \in   [0,1]}\sup_{\theta  \in  [\delta,1-\delta]}\sup_{f  \in
  \Sigma(\beta,L)} \esp_{\theta,f}(|\hat{f}_n^{\text{rwk}}(x)-f(x)|^{2}) &\leq 
 & C_1 \sup_{\theta \in [\delta,1-\delta]}\sup_{f \in \Sigma(\beta,L)}
\left[\esp_{\theta,f}\left(|\hat{\theta}_n-\theta|\right)^4\right]^{\frac{1}{2}}\\
& & +C_2 n^{\frac{-2\beta}{2\beta+1}},
\end{eqnarray*}
where  $C_1,C_2$   are  two  positive  constants   depending  only  on
$\beta, L, \alpha, \delta$ and $K$.
\end{thm}

The proof of this theorem is postponed to Section~\ref{ann:thm1}. It works as follows: 
we first start by proving that the pointwise quadratic risk of
$f_2$ (which is not an estimator) is of order $n^{-2\beta/(2\beta + 1)}$. Then we
compare   estimator $\hat{f}_n^{\text{rwk}}$ with  function $f_2$  to conclude
the  proof. We evidently  obtain  the  following
corollary from this theorem.

\begin{cor}\label{cor1}
Under  the assumptions of  Theorem~\ref{thm:cv_f}, if  $\hat \theta_n$
is such  that 
\begin{equation}
  \label{eq:condition}
\limsup_{n \rightarrow +\infty} n^{\frac{2\beta}{2\beta+1}}
\left[\esp_{\theta,f}\left(|\hat{\theta}_n-\theta|\right)^4\right]^{\frac{1}{2}} < +\infty,
\end{equation}
then for any fixed value $(\theta,f)$, there is some positive constant $C$ such that 
\begin{align*}
\sup_{x  \in   [0,1]}\esp_{\theta,f}(|\hat{f}_n^{\text{rwk}}(x)-f(x)|^{2}) \leq 
C n^{\frac{-2\beta}{2\beta+1}}. 
\end{align*}
\end{cor}

Note   that  estimators   $\hat  \theta_n$ satisfying~\eqref{eq:condition} exist. 
Indeed, relying on the same arguments as in the proofs of Propositions 2 or 3 in \cite{Prop_True_Null}, we can prove that for instance histogram-based estimators or the estimator proposed by~\cite{Celisse-Robin2010} both satisfy that 
 $$\limsup_{n \rightarrow +\infty} n
\left[\esp_{\theta,f}\left(|\hat{\theta}_n-\theta|\right)^4\right]^{\frac{1}{2}} < +\infty.$$
Note  also  that the  rate  $n^{-\beta/(2\beta  +  1)}$ is  the  usual
nonparametric minimax rate over  the class $\Sigma(\beta,L)$ of Hölder
densities  in  the case  of  direct  observations.   While we  do  not
formally   prove   that   this   is   also  the   case   in   undirect
model~\eqref{eq-model-pvalue}, it is likely that the rate
  in  this  latter  case  is   not  faster  as  the  problem  is  more
  difficult. A difficulty  in establishing such a lower  bound lies in
  the fact  that when  $\theta\in [\delta,1-\delta]$ the  direct model
  ($\theta=0$) is not a submodel of~\eqref{eq-model-pvalue}.
Anyway, such a lower bound would not be sufficient to conclude that estimator $\hat{f}_n^{\text{rwk}}$
achieves  the  minimax rate.  Indeed,  the corollary states nothing about uniform
convergence  of   $\hat  f_n^{\text{rwk}}(x)$  with   respect  to  the
parameter value  $(\theta,f)$ since  the convergence of  the estimator
$\hat \theta_n$ is not known to  be uniform.

\subsection{Maximum  smoothed likelihood estimator}\label{sec:msl}
Let  us now  explain  the motivations  for  considering an  iterative
procedure   with  functions   $\tilde  K_{i,h}$   and   weights  $\hat
\omega_i^{(s)}$ respectively defined through~\eqref{eq:tildeKih} and~\eqref{eq:weight_msl}.
Instead  of   the  classical  log-likelihood,  we   follow  the  lines
of~\cite{Levine} and consider (the opposite of)
a smoothed version of this log-likelihood as our criterion, namely 
\begin{equation*}
l_n(\theta,      f)     =      \frac{-1}{n}      \sum_{i=1}^n     \log
[\theta+(1-\theta)\mathcal{N}f(X_i)] .
\end{equation*}
In this section, we denote by $g_0$ the true density of the observations $X_i$. For any fixed value of
$\theta$, up to the additive constant $\int_0^1
g_0(x)\log g_0(x)dx$, the smoothed log-likelihood $l_n(\theta, f)$
converges almost surely towards $l(\theta, f)$
defined as
\begin{equation*}
l(\theta, f) :=\int_0^1 g_0(x) \log \frac{g_0(x)}{\theta + (1-\theta)\mathcal{N}f(x)}dx.
\end{equation*}
This quantity may be viewed as a penalized  Kullback-Leibler
divergence between the true density $g_0$ and its smoothed approximation for parameters $(\theta,f)$. Indeed, let $D(a\mid b)$ denote the  Kullback-Leibler divergence between (positive) measures $a$ and $b$, defined as
\begin{equation*}
D(a\mid b) = \int_0^1 \Big\{a(x)\log \frac{a(x)}{b(x)}+b(x)-a(x)\Big\}dx.
\end{equation*}
Note that in the above definition, $a$ and $b$ are not necessarily probability measures. Moreover it can be seen that we still have the property $D(a|b)\ge 0$ with equality if and only if $a=b$ \citep{Eggermont99}.
We now obtain 
\begin{equation*}
l(\theta, f) = D(g_0 | \theta + (1-\theta)\mathcal{N}f) + (1-\theta) \big(1-\int_0^1 \mathcal{N}f(x) dx \big).
\end{equation*}
The second term in the right-hand side of the above equation acts as a penalization term \citep{Eggermont99,Levine}. Our goal is to construct an iterative sequence of estimators of $f$ that possesses a descent property with respect
to the criterion $l(\theta,\cdot)$, for fixed value $\theta$. Indeed, as previously explained,  $\theta$ has to remain fixed otherwise the following procedure gives a sequence $\{\theta^t\}$ that converges to $0$.
We start by describing such a procedure, relying on the knowledge of the parameters (thus an oracle procedure).  Let us denote by $l_n(f)$ the smoothed log-likelihood
$l_n(\theta, f)$ and by $l(f)$ the limit function $l(\theta,f)$. We
want to construct  a sequence of densities $\{f^t\}_{t \ge 0}$ such that
\begin{equation}
\label{eq:creasing}
l(f^t) - l(f^{t+1}) \geq c D(f^{t+1}\mid f^t) \ge 0,
\end{equation}
where $c$ is a positive constant depending on $\theta$, the bandwidth $h$ and the kernel $K$. We thus consider the difference
\begin{eqnarray*}
l(f^t) - l(f^{t+1}) &=& \int_0^1 g_0(x)\log
\frac{\theta+(1-\theta)\mathcal{N}f^{t+1}(x)}{\theta+(1-\theta)\mathcal{N}f^{t}(x)}
dx\\
&=& \int_0^1 g_0(x)\log \Big\{1-\omega_t(x) +
  \omega_t(x)\frac{\mathcal{N}f^{t+1}(x)}{\mathcal{N}f^{t}(x)}\Big\} dx,
\end{eqnarray*}
where 
\begin{equation*}
\omega_t(x) = \frac{(1-\theta)\mathcal{N}f^t(x)}{\theta+(1-\theta)\mathcal{N}f^t(x)}. 
\end{equation*}
By the concavity of the logarithm function, we get that
\begin{eqnarray}
\label{eq-decroissante}
l(f^t) - l(f^{t+1}) &\geq& \int_0^1 g_0(x) \omega_t(x)\log
\frac{\mathcal{N}f^{t+1}(x)}{\mathcal{N}f^{t}(x)}dx \notag \\
&\geq& \int_0^1 g_0(x) \omega_t(x)\Big[\mathcal{S}^*(\log f^{t+1})(x) -
\mathcal{S}^*(\log f^t)(x)\Big] dx \notag \\
&\geq& \int_0^1 g_0(x) \omega_t(x) \big(\int_0^1 K_h(s-x)ds\big)^{-1}\Big(\int_0^1 K_h(u-x)\log
  \frac{f^{t+1}(u)}{f^t(u)}du\Big)dx \notag \\
&\geq& \int_0^1 \Big(\int_0^1\frac{ g_0(x)
  \omega_t(x)K_h(u-x)}{\int_0^1 K_h(s-x)ds} dx \Big)\log
  \frac{f^{t+1}(u)}{f^t(u)}du.
\end{eqnarray}
Let us define
\begin{equation}
\label{eq:f^t}
\alpha_t = \frac{1}{\int_0^1 \omega_t(u)g_0(u)du}\ \text{and} \ f^{t+1}(x) = \alpha_t\int_0^1\frac{
  K_h(u-x)\omega_t(u)g_0(u)}{\int_0^1 K_h(s-u)ds} du, 
\end{equation}
then $f^{t+1}$ is a density function on $[0,1]$ and
\begin{equation*}
l(f^t) - l(f^{t+1}) \geq \frac{1}{\alpha_t} D(f^{t+1}\mid f^t).
\end{equation*}
With the same arguments as in the proof of following Proposition~\ref{prop:hatft_decreases}, we can show that $\alpha_t^{-1}$ is lower bounded by a positive constant $c$ depending on $\theta, h$ and $K$. The
sequence        $\{f^t\}_{t\ge        0}$        thus        satisfies
property~\eqref{eq:creasing}. However, we stress that it is an oracle as it
depends on the knowledge of the true density $g_0$ that is unknown. 
Now, the estimator sequence $\{\hat f^{(t)}\}_{t\ge 0}$ defined through Equations~\eqref{eq:tildeKih},~\eqref{eq:weight_msl} and Algorithm~\ref{algo:iterative} is exactly the Monte Carlo approximation of   $\{f^t\}_{t\ge        0}$. 
We prove in the next proposition that it also satisfies the descent property~\eqref{eq:creasing}.

\begin{prop}\label{prop:hatft_decreases}         
For any initial value of the weights $\hat \omega_0\in (0,1)^n$, the sequence of estimators $\{\hat{f}^{(t)}\}_{t\ge 0}$ defined through~\eqref{eq:tildeKih},~\eqref{eq:weight_msl} and Algorithm~\ref{algo:iterative} satisfies
\begin{equation*}
l_n(\hat{f}^{(t)}) - l_n(\hat{f}^{(t+1)}) \geq c D(\hat{f}^{(t+1)}\mid \hat{f}^{(t)})\ge 0,
\end{equation*}
where $c$ is a positive constant depending  on $\theta$, the bandwidth $h$ and the kernel $K$.
\end{prop}

To conclude this section, we study the behavior of the limiting criterion $l$.
Let us  introduce the set 
\[
\mathcal{B}=\{\mathcal{S} \varphi ; \varphi \ \mbox{density on } [0,1]\}. 
\]

\begin{prop}\label{prop:subsequence}
The criterion $l$ has a unique minimum $f^\star$ on $\mathcal{B}$. Moreover, if there exists a constant $L$ depending on $h$ such that for all $x, y \in [-1,1]$
\begin{equation*}
| K_h(x) -K_h(y) | \leq L |x-y|,
\end{equation*}
then the sequence of densities $\{f^t\}_{t\ge 0}$ 
%has a subsequence which
converges uniformly to  $f^\star$. 
%a fixed point of $G$.
\end{prop}

Note that the  previous assumption may be satisfied  by many different
kernels. For  instance, if $K$ is  the density of  the standard normal
distribution, then this assumption is satisfied with 
\begin{equation*}
L=\frac{1}{h^2\sqrt{2\pi}} e^{-1/2}.
\end{equation*}

As a consequence and since $l_n$ is lower bounded, the sequence $\{\hat f^{(t)}\}_{t\ge 0}$ converges to a local minimum of $l_n$ as $t$ increases. Moreover, we recall that as the sample size $n$ increases, the criterion $l_n$ converges (up to a constant) to $l$. Thus, the outcome of Algorithm~\ref{algo:iterative} that relies on Equations~\eqref{eq:tildeKih} and~\eqref{eq:weight_msl} is an approximation of the minimizer $f^\star$ of $l$.

%%%%%%
%%%
\section{Estimation of local false discovery rate and simulation
  study}\label{sec:lfdr}
\subsection{Estimation of local false discovery rate}
In this section, we study the estimation of local false
discovery rate  (\lfdr) by using the  previously introduced estimators
of the density $f$ and compare these different approaches on simulated
data. Let us recall definition~\eqref{eq-lfdr} of the local false
discovery rate
\begin{equation*}
\lfdr(x) = \mathbb{P}( H_i\ \text{being true}\ | X=x) =
\frac{\theta }{\theta +(1-\theta)f(x)}, \quad x \in [0,1].
\end{equation*}
For a given estimator  $\hat{\theta}$ of the proportion $\theta$ and
an estimator $\hat{f}$ of the density $f$, we obtain a natural
estimator of the local false discovery rate for observation $x_i$  
\begin{equation}\label{eq:estim_lfdr}
\widehat{\lfdr}(x_i) =\frac{\hat{\theta}}{\hat{\theta} +(1-\hat{\theta})\hat{f}(x_i)}.
\end{equation}
Let us now denote by $\hat{f}_{\text{rwk}}$ the randomly weighted kernel estimator of $f$ constructed in
Section~\ref{sec:direct},  by  $\hat{f}_{\text{kerfdr}}$   the  estimator  of  $f$
presented  in Algorithm~\ref{algo:iterative}  and  by $\hat{f}_{\text{msl}}$  the maximum 
smoothed   likelihood     estimator   of   $f$   presented   in
Algorithm~\ref{algo:iterative}. Note that $\hat{f}_{\text{kerfdr}}$  is available through the R package \texttt{kerfdr}.  
 We also  let $\widehat{\lfdr}_m,  m \in
\{ \text{rwk}, \text{kerfdr},\text{msl}\}$ be the
 estimators  of $\lfdr$  induced by  a plug-in of  estimators $\hat
 f_m$ in~\eqref{eq:estim_lfdr} and $\widehat{\lfdr}_{st}$ be the estimator  of $\lfdr$ computed by the method of \cite{Strimmer2008}. We
compute the root mean squared error (RMSE) between the
estimates and the true values
\begin{equation*}
\mbox{RMSE}_m
=\frac{1}{S}\sum_{s=1}^S\sqrt{\frac{1}{n}\sum_{i=1}^n\{ \widehat{\lfdr}_m^{(s)}(x_i)-\lfdr(x_i)\}^2},
\end{equation*}
for $m \in \{ \text{rwk}, \text{kerfdr},\text{msl}, \text{st}\}$ and where
$s=1,\ldots,S$ denotes  the simulation index ($S$ being the total
number of repeats). We also compare $\mathbb{L}^2$-norms between $\hat{f}_m$ and $f$ for $m \in \{ \text{rwk}, \text{kerfdr},\text{msl}\}$, relying on the root mean integrated squared error 
\begin{equation*}
\mbox{RMISE}_m
=\frac{1}{S}\sum_{s=1}^S\sqrt{\int_0^1 [\hat{f}_m^{(s)}(u)-f(u)]^2du}.
\end{equation*}
 The quality of the estimates provided by method $m$
is measured by the mean $\mbox{RMSE}_m$ or $\mbox{RMISE}_m$: the smaller these quantities,  the
better the performances of the method.

We mention that we also tested the naive method described in Section~\ref{sec:direct} and the results were bad. In order to present clear figures, we have chosen not to show those.

\subsection{Simulation study}
In this section, we give an illustration of the previous results on
some simulated experiments. We simulate sets of $p$-values according
to the mixture model~\eqref{eq-model-pvalue}. We consider  three
different cases for the alternative distribution $f$ and two
different values for the proportion: $\theta = 0.65$ and $0.85$. In the first case, we simulate $p$-values under the 
alternative with distribution 
\begin{equation*}
f(x)=\rho \Big(1-x\Big)^{\rho-1}\mathbf{1}_{[0,1]}(x),
\end{equation*}
where $\rho =4$, as proposed in \cite{Celisse-Robin2010}. In the second case, the $p$-value corresponds
to the statistic $T$ which has a mixture distribution $\theta
\mathcal{N}(0,1) + (1-\theta) \mathcal{N} (\mu, 1)$, with $\mu=2$. In the third case, the $p$-value corresponds to the statistic $T$ which has a mixture density $\theta
(1/2)\exp\{-|t|\} + (1-\theta) (1/2)\exp\{-|t-\mu|\}$, with
$\mu=1$. The $p$-values densities obtained with those three models are given in Figure~\ref{fig:densities} for $\theta=0.65$.

\begin{figure}[htbp]
  \centering
  \includegraphics[width=0.8\columnwidth]{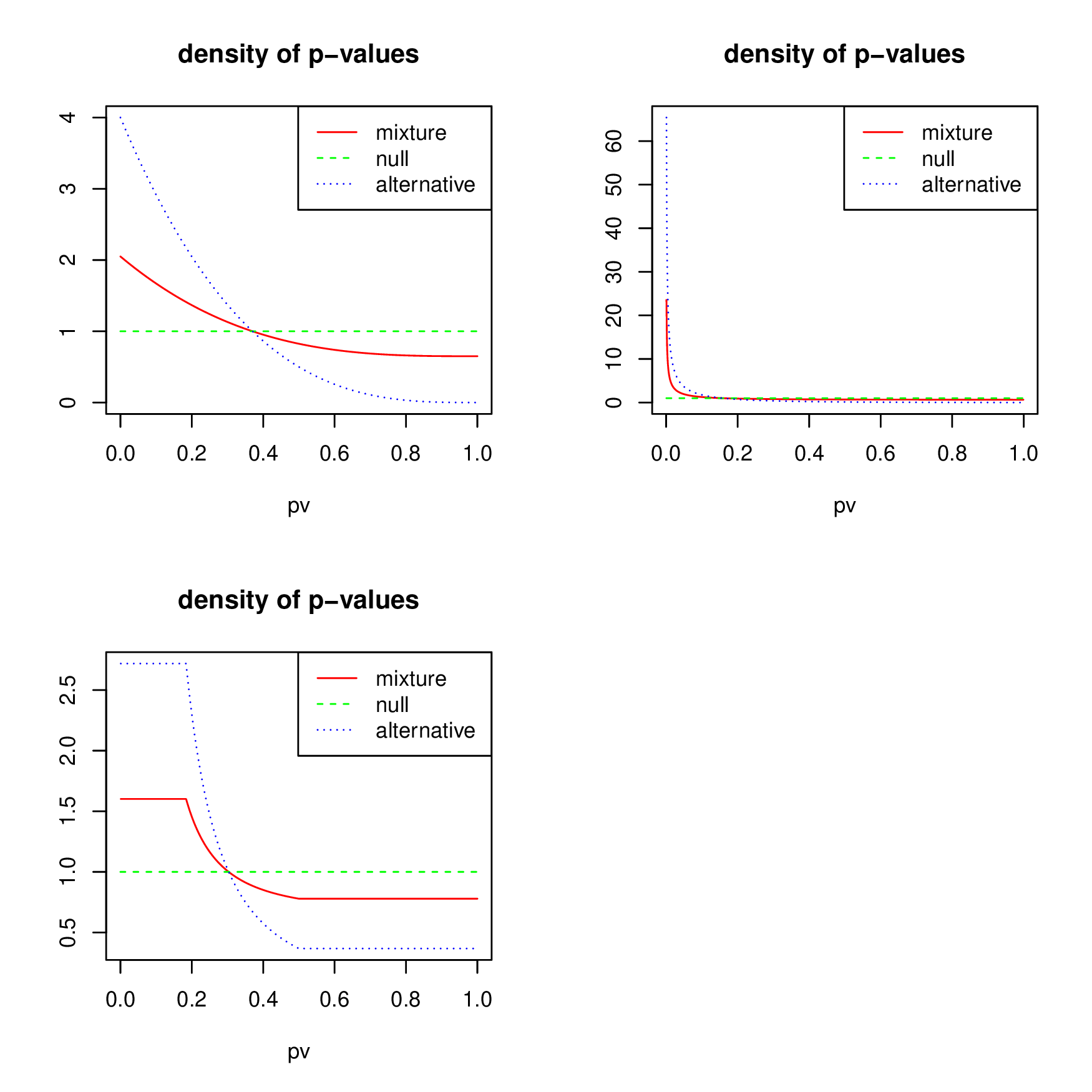}
  \caption{Densities of the $p$-values in the three different models, with $\theta=0.65$.  Top left: first model, top right: second model, bottom left: third model.}
  \label{fig:densities}
\end{figure}

For each of the $3\times 2 =6$ configurations, we generate
$S =100$ samples of size $n \in \{500, 1000, 2000, 5000 \}$. In these experiments, we choose to consider the estimator of $\theta$ initially proposed by \cite{Schweder1982}, namely
\begin{equation*}
\hat{\theta} = \frac{\#\{X_i>\lambda; i=1,\ldots,n\}}{n(1-\lambda)},
\end{equation*}
with parameter value $\lambda$ optimally chosen by bootstrap method, as recommended by \cite{Storey2002}. The kernel is chosen with compact support, for example the triangular kernel or the rectangular kernel. The bandwidth is selected according to a rule of thumb due to \cite[][Section 3.4.2]{Silverman1986}, 
\begin{equation*}
h = 0.9 \min \Big\{SD, \frac{IQR}{1.34}\Big\} n^{-1/5},
\end{equation*}
where $SD$ and $IQR$ are respectively the standard deviation and interquartile range of the data values. Figures~\ref{fig:illustration1}, ~\ref{fig:illustration2} and ~\ref{fig:illustration3}  show the RMISEs and the RMSEs for the six configurations and the four different methods.

\begin{figure}[htbp]
\begin{center}
\includegraphics[width=0.8\columnwidth]{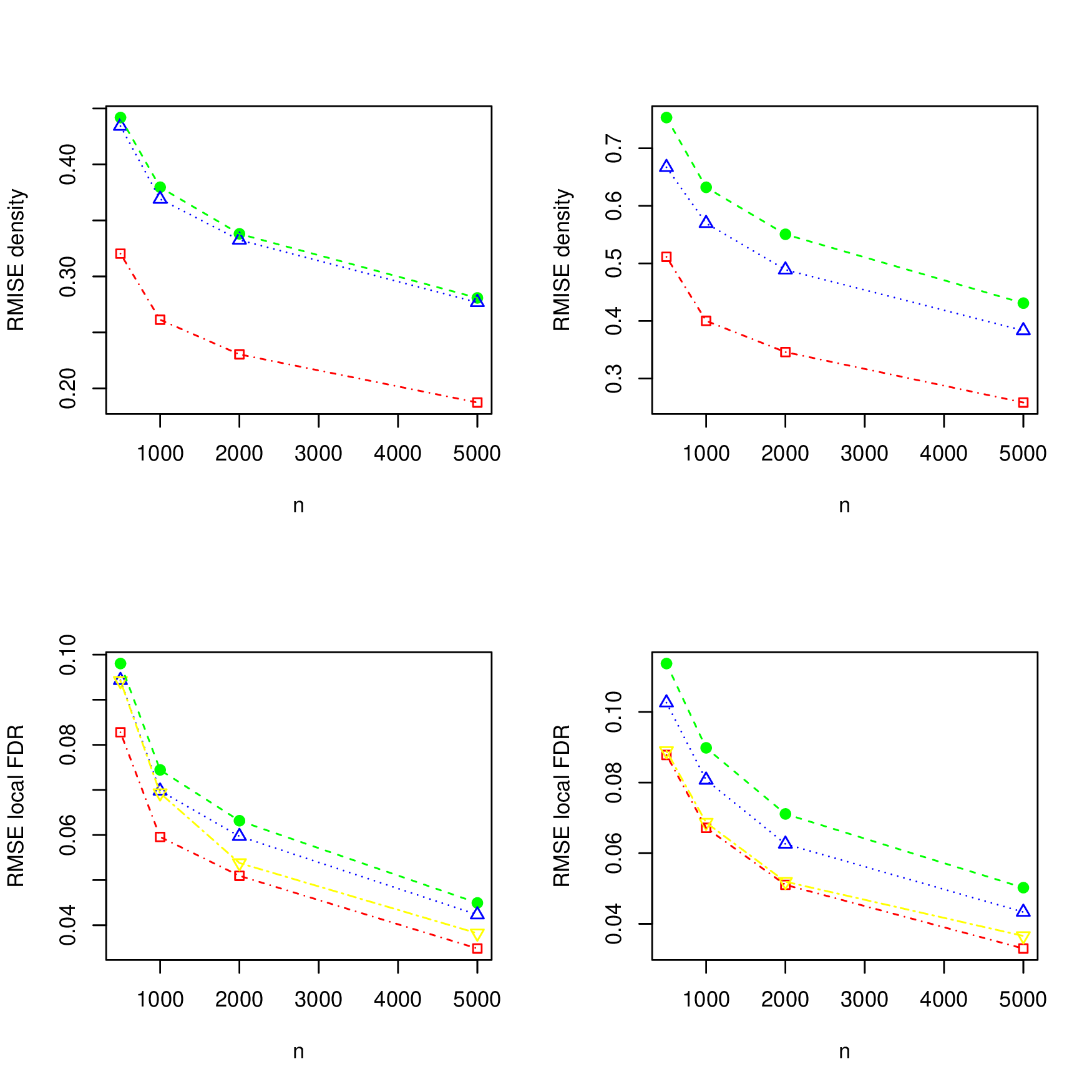}
\caption{RMISE (for density $f$) and RMSE (for \lfdr) in the first model as a function of $n$. Methods: "$\bullet$" =  rwk, "$\triangle$" =  \texttt{kerfdr},  "$\Box$" = msl, "$\triangledown$" = st (only for \lfdr). Left: $\theta = 0.65$, right: $\theta = 0.85$.}
\label{fig:illustration1}
\end{center}
\end{figure}

\begin{figure}[htbp]
\begin{center}
\includegraphics[width=0.8\columnwidth]{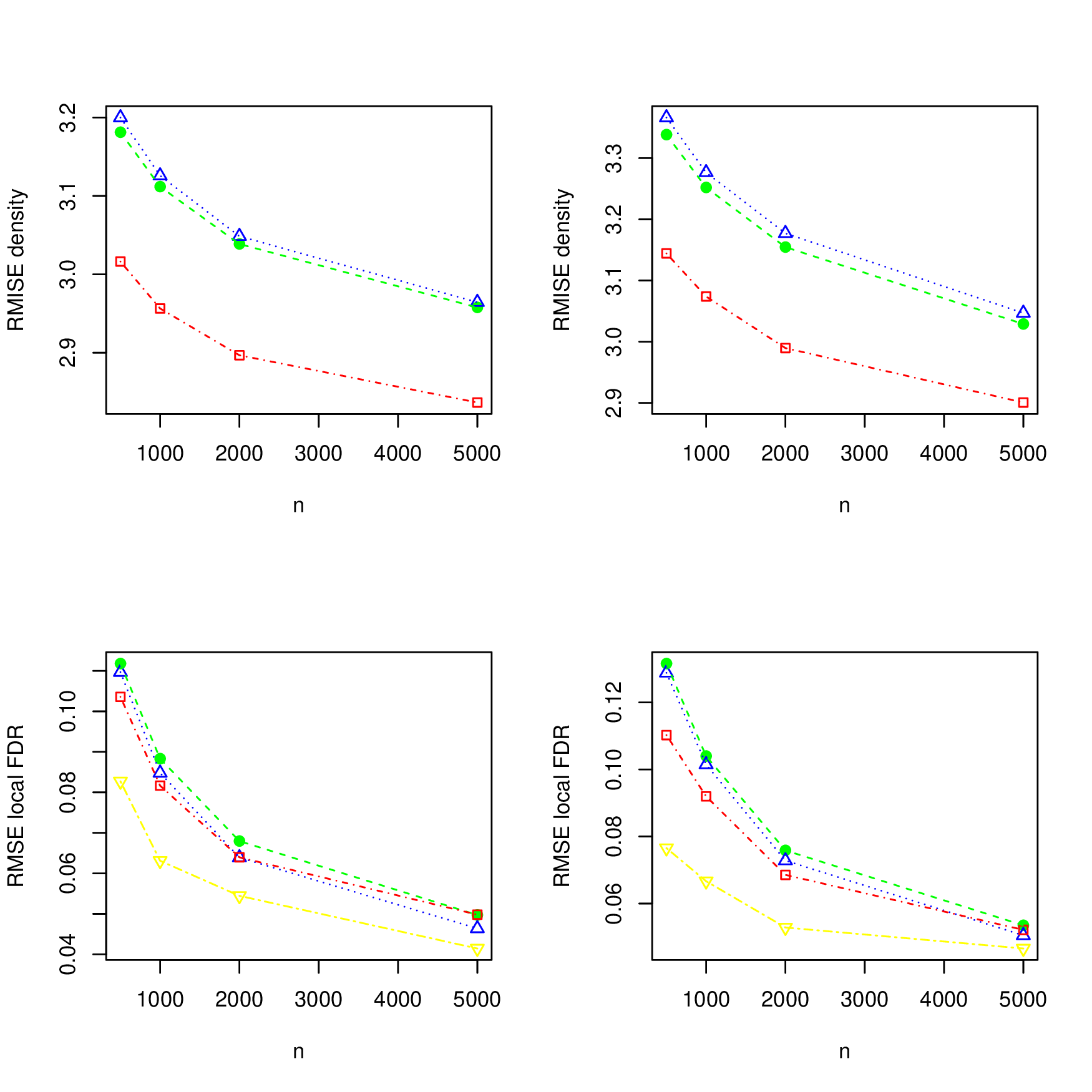}
\caption{RMISE (for density $f$) and RMSE (for \lfdr) in the second model as a function of $n$. Methods: "$\bullet$" =  rwk, "$\triangle$" =  \texttt{kerfdr},  "$\Box$" = msl,  "$\triangledown$" = st (only for \lfdr). Left: $\theta = 0.65$, right: $\theta = 0.85$.}
\label{fig:illustration2}
\end{center}
\end{figure}

\begin{figure}[htbp]
\begin{center}
\includegraphics[width=0.8\columnwidth]{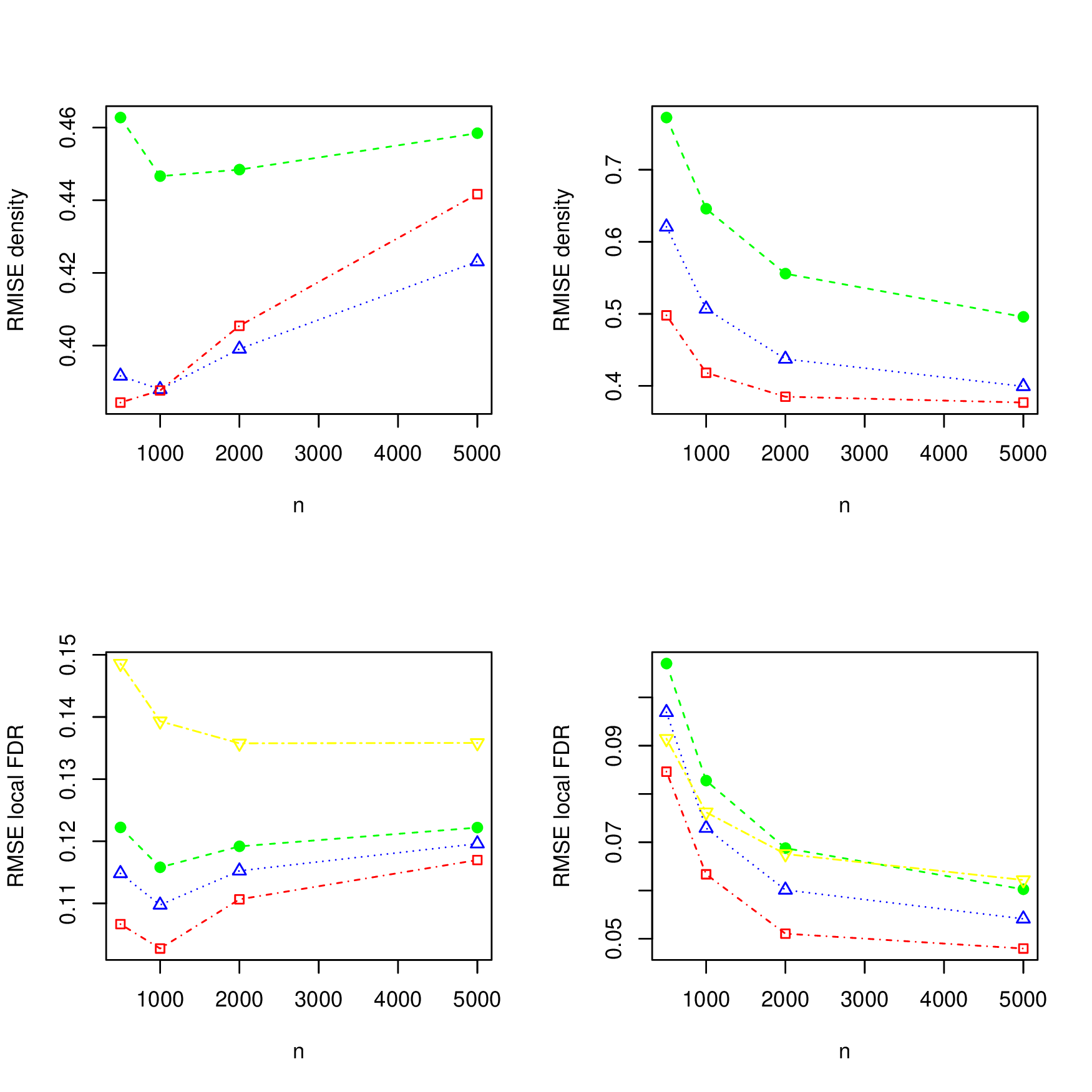}
\caption{RMISE (for density $f$) and RMSE (for \lfdr)  in the third model as a function of $n$. Methods: "$\bullet$" =  rwk, "$\triangle$" =  \texttt{kerfdr}, "$\Box$" = msl, "$\triangledown$" = st  (only for \lfdr). Left: $\theta = 0.65$, right: $\theta = 0.85$.}
\label{fig:illustration3}
\end{center}
\end{figure}

We first comment the results on the estimation of $f$ (top half of each figure). Except for model 2, the RMISEs obtained are small for all the three procedures. Model 2 exhibits a rather high RMISEs and this may be explained by the fact that density $f$ is not bounded near $0$ in this case.  
We note that the methods \texttt{rwk} and \texttt{kerfdr} have very similar performances, except in the third model where \texttt{kerfdr} seems to slightly outperform \texttt{rwk}. Let us recall that we introduced this latter method only as a way of approaching the theoretical performances of \texttt{kerfdr} method.
Now, in five out of the six configurations, \texttt{msl} outperforms the two other methods (\texttt{rwk, kerfdr}).

Then, we switch to comparing the methods with respect to estimation of \lfdr\ (bottom half of each figure). 
First, note  that the four methods  exhibit small RMSEs with respect to \lfdr\  and are thus efficient for estimating this quantity. We also note that \texttt{rwk} tends to have lower performances than \texttt{kerfdr,msl}.  Now, \texttt{msl} tends to slightly outperform \texttt{kerfdr}. Thus \texttt{msl} appears as a competitive method for \lfdr\ estimation. The comparison with \cite{Strimmer2008}'s approach is more difficult: for model 1, the method compares with \texttt{msl}, while it outperforms all the methods in model 2 and is outperformed by \texttt{msl} in model 3. 

As a conclusion, we claim that \texttt{msl} is a competitive method for estimating both the alternative density $f$ and the \lfdr.

%%%%%%%%%%
%%%%%%%%%%
\section{Proofs}\label{sec:proofs}
\subsection{Proof of Theorem~\ref{thm:cv_f}}\label{ann:thm1}
The proof works as follows: 
we  first  start by  proving  that  the  pointwise quadratic  risk  of
function $f_2$ defined by~\eqref{eq:f2} is order of $n^{-2\beta/(2\beta + 1)}$ in the
following proposition. Then we
compare the  estimator $\hat{f}_n^{\text{rwk}}$ with  the function $f_2$  to conclude
the     proof.     To     simplify     notation,     we     abbreviate
$\hat{f}_n^{\text{rwk}}$ to $\hat f_n$. 

We shall need the following two lemmas. The proof of the first one may be found for
instance in Proposition 1.2 in \cite{Tsybakov-book}. The second one is
known as Bochner's  lemma and is a classical  result in kernel density
estimation. Therefore its proof is omitted. 
\begin{lem} \textbf{(Proposition 1.2 in \cite{Tsybakov-book}).}
\label{lem1}
Let $p$ be a density in $\Sigma(\beta,L)$ and $K$ a kernel function of order $l
= \lfloor \beta \rfloor$ such that 
\begin{equation*}
\int_{\mathbb{R}}| u| ^{\beta}| K(u)| du < \infty .
\end{equation*}
Then there exists a positive constant $C_3$ depending only on $\beta, L$ and $K$ such that for all $x_0 \in \mathbb{R}$,
\begin{equation*}
\Big|\int_{\mathbb{R}}K(u)\big[p(x_0 + uh) - p(x_0)\big] du\Big| \leq
C_3h^{\beta},\quad \forall h > 0.
\end{equation*}
\end{lem} 
%%%
%%%
\begin{lem} 
\label{lem:Bochner}
\textbf{(Bochner's lemma).} Let $g$ be a bounded function on
$\mathbb{R}$, continuous in a neighborhood of $x_0 \in \mathbb{R}$ and
$Q$ a function which satisfies
\begin{equation*}
\int_{\mathbb{R}}| Q(x)| dx < \infty .
\end{equation*}
Then, we have
\begin{equation*}
\displaystyle\lim_{h\to      0}\frac      {1}     {h}\int_{\mathbb{R}}
Q\Big( \frac{x - x_0}{h}\Big) g(x)dx
= g(x_0)\int_{\mathbb{R}}Q(x)dx .
\end{equation*}
\end{lem}

Now, we come to the first step in the proof.
\begin{prop}
\label{prop:cv_f2}
Assume that kernel $K$ satisfies Assumption~\hyperref[hyp:A3]{\textbf{(A3)}} and  bandwidth $h=\alpha n^{-1/(2\beta + 1)}$, with
$\alpha > 0$. Then the
pointwise quadratic risk of function $f_2$, defined by~\eqref{eq:f2}
and depending on $(\theta,f)$, satisfies 
\begin{equation*}
\sup_{x \in [0,1]} \sup_{\theta\in [\delta,1-\delta]}\sup_{f \in \Sigma(\beta,L)} \esp_{\theta,f}(|f_2(x)-f(x)|^{2}) \leq C_4 n^{\frac{-2\beta}{2\beta + 1}},
\end{equation*}
where  $C_4$ is a positive  constant   depending  only  on
$\beta, L, \alpha, \delta$ and $K$.
\end{prop}

\begin{proof}[Proof of Proposition~\ref{prop:cv_f2}]
Let us denote by 
\begin{align*}
S_n = \sum_{i=1}^n \frac {f(X_i)} {g(X_i)} . 
%\quad \text{ and } \forall j\in \{1,\dots,n\}, \ S_{n,j} = \sum_{ i\neq j} \frac {f(X_i)} {g(X_i)}.
\end{align*}
The pointwise quadratic risk of $f_2$ can be written as the
sum of a bias term and a variance term
\begin{equation*}
\esp_{\theta,f}(|f_2(x)-f(x)|^{2}) = [\esp_{\theta,f}(f_2(x))-f(x)]^2 + \Var_{\theta,f}[f_2(x)].
\end{equation*} 
Let us first study the bias term. According to~\eqref{eq:f2} and the definition~\eqref{eq:posterior-proba} of the weights, we have
\begin{eqnarray}
\label{eq-esp-f2}
\esp_{\theta,f}[f_2(x)] &=& \frac{n}{h}
\esp_{\theta,f}\left[\tau_1K\Big(\frac{x-X_1}{h}\Big)\Big(\sum_{k=1}^n \tau_k\Big)^{-1}\right] \notag
\\
&=& \frac{n}{h} \esp_{\theta,f}\left[\frac{f(X_1)}{g(X_1)}K\Big(\frac{x-X_1}{h}\Big)S_n^{-1}\right] \notag \\
&=& \frac{n}{h} \int_0^1f(t)K\Big(\frac {x - t} {h}\Big) \esp_{\theta,f}\left[\Big(\frac
  {f(t)} {g(t)} + S_{n-1}\Big)^{-1}\right]dt \notag
  \\
&=& n \int_{-x/h}^{(1-x)/h}K(t)f(x + th) \esp_{\theta,f}\left[\Big(\frac {f(x + th)} {g(x +
    th)} + S_{n-1}\Big)^{-1}\right]dt .
\end{eqnarray}
Since the functions $f$ and $g$ are related by the equation $g(t) =
\theta+(1-\theta)f(t)$ for all $t \in [0,1]$, the ratio $f(t)/g(t)$ is well defined and satisfies 
\begin{align*}
0 \leq \frac {f(t)} {g(t)} \leq \frac{1}{1 - \theta} \le \delta^{-1} , \quad \forall t
\in [0,1], \mbox{ and } \forall \theta \in [\delta, 1-\delta].
\end{align*}
Then for all $t \in [-x/h,(1-x)/h]$, we get 
\begin{eqnarray*}
\frac {1} {S_{n-1} + \delta^{-1}} \leq
\left(\frac {f(x + th)} {g(x + th)} + S_{n-1}\right)^{-1} \leq \frac {1} {S_{n- 1}},
\end{eqnarray*}
where the bounds are uniform with respect to $t$.\\
By combining this inequality with \eqref{eq-esp-f2}, we obtain 
\begin{align*}
&n\left(\int_{-x/h}^{(1-x)/h}K(t)f(x+th)dt\right)\esp_{\theta,f}\Big(\frac
{1} {S_{n- 1} + \delta^{-1}}\Big) \leq \esp_{\theta,f}\big[f_2(x)\big] \\
\mbox{and } \quad 
&\esp_{\theta,f}\big[f_2(x)\big] \leq n\left(\int_{-x/h}^{(1-x)/h}K(t)f(x+th)dt\right)\esp_{\theta,f}\Big(\frac {1} {S_{n- 1}}\Big).
\end{align*}
Then,  we apply  the  following  lemma, whose  proof  is postponed  to
Appendix~\ref{sec:appen_Sn}.
\begin{lem}
\label{lem-expectation-Sn}
There exist  some positive constants  $c_1, c_2, c_3,  c_4$ (depending
on $\delta$) such that for $n$
large enough,
\begin{align}
&\esp_{\theta,f}\Big(\frac {1} {S_n}\Big) \leq \frac{1}{n} + \frac{c_1}{n^2}\label{eq:toto9},\\
&  \esp_{\theta,f}\Big(\frac
{1} {S_n^2}\Big) \leq \frac{c_2}{n^2},\label{eq:toto10}\\
& \esp_{\theta,f}\Big(\frac {1} {S_n+2\delta^{-1}}\Big) \geq  \frac{1}{n} - \frac{c_3}{n^2}\label{eq:toto11},\\
\mbox{and} \quad & \esp_{\theta,f}\Big(\frac {1} {S_n^2}\Big) - \esp_{\theta,f}^2\Big(\frac {1} {\delta^{-1} + S_n}\Big) \leq \frac{c_4}{n^3}. \label{eq:toto12}
\end{align}
\end{lem}
Relying on Inequalities~\eqref{eq:toto9} and~\eqref{eq:toto11}, we have for $n$ large enough
\begin{eqnarray*}
\int_{-x/h}^{(1-x)/h}K(t)f(x+th)dt -\frac {c_3} {n} \leq
\esp_{\theta,f}\big[f_2(x)\big] \leq \int_{-x/h}^{(1-x)/h}K(t)f(x+th)dt + \frac {c_1} {n}.
\end{eqnarray*}
Since $f(x+th) = 0$ for all $t \notin [-x/h,(1-x)/h]$, we may write 
\begin{equation*}
\int_{-x/h}^{(1-x)/h}K(t)f(x+th)dt =
\int_{\mathbb{R}}K(t)f(x+th)dt .
\end{equation*}
Thus,
% we have
% \begin{eqnarray*}
% \int_{\mathbb{R}}K(t)f(x+th)dt -\frac {c_3} {n} \leq
% \esp_{\theta,f}\big[f_2(x)\big] \leq \int_{\mathbb{R}}K(t)f(x+th)dt + \frac {c_1} {n},
% \end{eqnarray*}
% and 
the bias of $f_2(x)$ satisfies
\begin{equation*}
|b(x)| = |\esp_{\theta,f}\big[f_2(x)\big] -f(x)|\leq \int_{\mathbb{R}}K(t)|f(x + th) - f(x)|dt + \frac {c_5} {n}.
\end{equation*}
By using Lemma~\ref{lem1} and the choice of bandwidth $h$, we obtain that 
\begin{equation*}
b^2(x) \leq  C_5 h^{2\beta},
\end{equation*}
where $C_5 = C_5(\beta, L, K)$. Let us study now the variance term of $f_2(x)$. We have 
\begin{equation}
\label{eq:toto13}
\Var_{\theta,f}\big[f_2(x)\big] = \frac{1}{h^2} \big[n \Var_{\theta,f}(Y_1) + n(n-1) \Cov_{\theta,f}(Y_1,Y_2)\big],
\end{equation}
where
\begin{equation*}
Y_i = \frac{f(X_i)} {g(X_i)} K\Big(\frac{x-X_i}{h}\Big)S_n^{-1}.
\end{equation*}
The variance of $Y_1$ is bounded by its second moment and 
%: $\Var_{\theta,f}(Y_1) \leq \esp_{\theta,f}(Y_1^2)$.
%in the following way
%\begin{equation*}
%\Var_{\theta,f}(Y_1) \leq \esp_{\theta,f}(Y_1^2).
%\end{equation*}
%Moreover,
\begin{multline*}
\esp_{\theta,f}(Y_1^2) = \esp_{\theta,f}\left[\Big(\frac{f(X_1)} {g(X_1)}\Big)^2
K^2\Big(\frac{x-X_1}{h}\Big)S_n^{-2}\right]  \\
= \int_0^1\frac{f^2(t)} {g(t)} K^2\Big(\frac{x-t}{h}\Big)
\esp_{\theta,f}\Big[\Big(\frac {f(t)} {g(t)} + S_{n-1}\Big)^{-2}\Big] dt .
\end{multline*}
Now,   recalling   that  $0\le   f/g   \le   \delta^{-1}$  and   using
Inequality~\eqref{eq:toto10} of Lemma~\ref{lem-expectation-Sn}, we get 
\begin{eqnarray}
\label{eq:toto14}
\esp_{\theta,f}(Y_1^2)
&\leq& h\Big(\int_{-x/h}^{(1-x)/h}\frac{f^2(x + th)} {g(x + th)} K^2(t)dt\Big) \esp_{\theta,f}\Big(\frac {1}{S_{n - 1}^2}\Big) \notag \\
&\leq& h \delta^{-1} \sup_{f \in \Sigma(\beta,L)}\|f\|_{\infty}\Big(\int K^2(t)dt\Big) \frac {c_2}{n^2} \leq \frac {C_6h} {n^2} .
\end{eqnarray}
We now study the covariance of $Y_1$ and $Y_2$
\begin{eqnarray*}
& &\Cov_{\theta,f}(Y_1,Y_2) =  \esp_{\theta,f}(Y_1Y_2) - \esp_{\theta,f}^2(Y_1)\\
 &=& \esp_{\theta,f}\left[\frac{f(X_1)f(X_2)} {g(X_1)g(X_2)}K\Big(\frac{x-X_1}{h}\Big)K\Big(\frac{x-X_2}{h}\Big)S_n^{-2}\right] - \esp_{\theta,f}^2\left[\frac{f(X_1)} {g(X_1)}
K\Big(\frac{x-X_1}{h}\Big)S_n^{-1}\right]\\
&=& \int_{[0,1]^2}f(t)f(u)K\Big(\frac{x-t}{h}\Big)K\Big(\frac{x-u}{h}\Big)
\esp_{\theta,f}\left[\Big(\frac {f(t)} {g(t)} + \frac {f(u)} {g(u)} + S_{n-2}\Big)^{-2}\right] dtdu\\
& & -  \left(\int_0^1f(t)K\Big(\frac{x-t}{h}\Big) \esp_{\theta,f}\left[\Big(\frac {f(t)} {g(t)}
+ S_{n-1}\Big)^{-1}\right] dt\right)^2 \\
&=& \int_{[0,1]^2}f(t)f(u)K\Big(\frac{x-t}{h}\Big)K\Big(\frac{x-u}{h}\Big)A(t,u)dtdu,
\end{eqnarray*}
where
\begin{eqnarray*}
A(t,u) &=& \esp_{\theta,f}\left[\Big(\frac {f(t)} {g(t)} + \frac {f(u)} {g(u)} +
S_{n-2}\Big)^{-2}\right] - \esp_{\theta,f}\left[\Big(\frac {f(t)} {g(t)} +
S_{n-1}\Big)^{-1}\right]\esp_{\theta,f}\left[\Big(\frac {f(u)} {g(u)} + S_{n-1}\Big)^{-1}\right]\\
&\leq& \esp_{\theta,f}\Big(\frac {1} {S_{n - 2}^2}\Big) - \esp_{\theta,f}^2\Big(\frac {1} {2\delta^{-1} + S_{n - 2}}\Big).
\end{eqnarray*}
Hence
\begin{eqnarray*}
\Cov(Y_1,Y_2) &\leq& \int_{[0,1]^2}f(t)f(u)K\Big(\frac{x-t}{h}\Big)K\Big(\frac{x-u}{h}\Big)\left[\esp_{\theta,f}\Big(\frac {1} {S_{n - 2}^2}\Big) - \esp_{\theta,f}^2\Big(\frac {1} {2\delta^{-1}  + S_{n - 2}}\Big)\right]dtdu\\
&\leq& h^2\left(\int_{\mathbb{R}}f(x + th)K(t)dt\right)^2\left[\esp_{\theta,f}\Big(\frac {1} {S_{n - 2}^2}\Big) -
\esp_{\theta,f}^2\Big(\frac {1} {2\delta^{-1}  + S_{n - 2}}\Big)\right]\\
&\leq&  C_7h^2\left[\esp_{\theta,f}\Big(\frac {1} {S_{n - 2}^2}\Big) - \esp_{\theta,f}^2\Big(\frac {1} {2\delta^{-1}  + S_{n - 2}}\Big)\right].\\
\end{eqnarray*}
According          to          Inequality~\eqref{eq:toto12}         of
Lemma~\ref{lem-expectation-Sn}, we have
\begin{equation*}
\esp_{\theta,f}\Big(\frac {1} {S_{n-2}^2}\Big) - \esp_{\theta,f}^2\Big(\frac {1} {2\delta^{-1} + S_{n-2}}\Big) \leq
\frac{c_4}{n^3},
\end{equation*}
hence
\begin{equation}
\label{eq:toto15}
\Cov_{\theta,f}(Y_1,Y_2) \leq \frac{C_8h^2}{n^3}.
\end{equation}
By returning to Equality~\eqref{eq:toto13} and combining with
\eqref{eq:toto14} and \eqref{eq:toto15}, we obtain
\begin{eqnarray*}
 \Var_{\theta,f}\big[f_2(x)\big] \leq \frac{1}{h^2}\left[\frac{C_6h}{n} + n(n -
 1)h^2\frac{C_8h^2}{n^3}\right] \leq \frac{C_9}{nh}.
\end{eqnarray*}
Thus, as the bandwidth $h$ is of order $n^{-1/(2\beta + 1)}$,
the pointwise quadratic risk of $f_2(x)$ satisfies
\begin{align*}
\esp_{\theta,f}(|f_2(x)-f(x)|^{2}) \leq C_4n^{\frac{-2\beta}{2\beta + 1}}.
\end{align*}
\end{proof}

\begin{proof}[Proof of Theorem~\ref{thm:cv_f}.]
First, the pointwise quadratic risk of $\hat{f}_n(x)$ is bounded in the
following way
\begin{equation}
\label{eq:toto20}
\esp_{\theta,f}(|\hat{f}_n(x)-f(x)|^{2}) \leq
2\esp_{\theta,f}(|f_2(x)-f(x)|^{2}) +
2\esp_{\theta,f}(|\hat{f}_n(x)-f_2(x)|^{2}).
\end{equation}
According to Proposition~\ref{prop:cv_f2}, we have 
\begin{equation}
\label{eq:toto21}
\esp_{\theta,f}(|f_2(x)-f(x)|^{2}) \leq C_4n^{\frac{-2\beta}{2\beta+1}},
\end{equation}
and it remains to study the second term appearing in the right-hand side of~\eqref{eq:toto20}. We write 
\begin{eqnarray*}
\hat{f}_n(x) - f_2(x) &=&
\frac{1}{h}\sum_{i=1}^n\left(\frac{\hat{\tau}_i}{\sum_k\hat{\tau}_k} -
\frac{\tau_i}{\sum_k\tau_k}\right)K\left(\frac{x-X_i}{h}\right)\\
%&=&
%\frac{1}{h}\sum_{i=1}^n\left(\frac{\hat{\tau}_i}{\sum_k\hat{\tau}_k}
%-\frac{\tau_i}{\sum_k\hat{\tau}_k} +
%\frac{\tau_i}{\sum_k\hat{\tau}_k} -\frac{\tau_i}{\sum_k\tau_k}\right)K\left(\frac{x-X_i}{h}\right)\\
&=& \frac{1}{h}\sum_{i=1}^n\frac{\hat{\tau}_i -
  \tau_i}{\sum_k\hat{\tau}_k}K\left(\frac{x-X_i}{h}\right) + \frac{1}{h}\sum_{i=1}^n\tau_i\left(\frac{1}{\sum_k\hat{\tau}_k} -
\frac{1}{\sum_k\tau_k}\right)K\left(\frac{x-X_i}{h}\right)\\
&=& \frac{n}{\sum_k\hat{\tau}_k} \times\frac{1}{n h}\sum_{i=1}^n(\hat{\tau}_i -
  \tau_i) K\left(\frac{x-X_i}{h}\right)\\
&&+ \frac{n^2}{\sum_k\hat{\tau}_k \sum_k\tau_k}\times
\frac{\sum_k(\tau_k-\hat{\tau}_k)}{n}\times\frac{1}{n h}\sum_{i=1}^n\tau_i K\left(\frac{x-X_i}{h}\right).
\end{eqnarray*}
Moreover, recalling the definition of the weights~\eqref{eq:def_weight}, we have for all $1\le i\le n$, 
\begin{equation*}
\hat{\tau}_i - \tau_i = \frac{\hat{\theta}_n}{\tilde{g}_n(X_i)}-\frac{\theta}{g(X_i)}
=\hat{\theta}_n\Big[\frac{1}{\tilde{g}_n(X_i)}-\frac{1}{g(X_i)}\Big]+\frac{1}{g(X_i)}(\hat{\theta}_n-\theta),
\end{equation*}
and thus  get 
\begin{eqnarray}
\label{eq:equality-difference}
\hat{f}_n(x) - f_2(x) &=& \frac{n\hat{\theta}_n}{\sum_k\hat{\tau}_k}
\times\frac{1}{n
  h}\sum_{i=1}^n\Big[\frac{1}{\tilde{g}_n(X_i)}-\frac{1}{g(X_i)}\Big]
K\left(\frac{x-X_i}{h}\right)\notag\\
&&+ \frac{n(\hat{\theta}_n-\theta)}{\sum_k\hat{\tau}_k} \times\frac{1}{n h}\sum_{i=1}^n\frac{1}{g(X_i)} K\left(\frac{x-X_i}{h}\right)\notag\\
&&+ \frac{n^2\hat{\theta}_n}{\sum_k\hat{\tau}_k \sum_k\tau_k}\times\frac{1}{n}
\sum_k\Big[\frac{1}{\tilde{g}_n(X_k)}-\frac{1}{g(X_k)}\Big]\times\frac{1}{n
  h}\sum_{i=1}^n\tau_i K\left(\frac{x-X_i}{h}\right)\notag\\
&&+ \frac{n^2(\hat{\theta}_n-\theta)}{\sum_k\hat{\tau}_k \sum_k\tau_k}\times\frac{1}{n}
\sum_k\frac{1}{g(X_k)}\times\frac{1}{n
  h}\sum_{i=1}^n\tau_i K\left(\frac{x-X_i}{h}\right).
\end{eqnarray}
Let us control the different terms appearing in this latter equality. 
We first remark that for all $i$, 
\begin{equation}
\label{eq:bound-tau-g}
0 \leq \tau_i\leq 1 \ \text{and}\ \frac{1}{g(X_i)} \leq \frac{1}{\theta} \le \delta^{-1}.
\end{equation}
Since by assumption $\hat{\theta}_n \xrightarrow[n\to\infty] {as} \theta \in [0,1]$,
for $n$ large enough we also get $|\hat{\theta}_n| < 3/2 , \mbox{a.s.}$ 
According to the law of large numbers and $\esp_{\theta,f}(\tau_1) =1 - \theta$, we also obtain that for $n$ large enough 
%\begin{equation*}
%\frac{1}{n}\sum_{i=1}^n\tau_i \xrightarrow[n\to\infty] {as} \esp_{\theta,f}(\tau_1) =1 - \theta.
%\end{equation*}
%Hence, 
\begin{equation}
\label{eq:bound-sum-tau}
\frac \delta 2 \leq \frac{1-\theta}{2} \leq \frac{1}{n}\sum_{i=1}^n\tau_i \leq \frac{3(1-\theta)}{2} \le \frac{3(1-\delta)}{2}
\quad \mbox{a.s.}
\end{equation}
Moreover, by using a Taylor expansion of the function $u\mapsto 1/u$ with an integral form of the remainder term, we have for all $i$,
\[
\Big|\frac{1}{\tilde{g}_n(X_i)}-\frac{1}{g(X_i)}\Big| = \frac{|\tilde{g}_n(X_i)-g(X_i)|}{g^2(X_i)}\int_0^1\left(1
    + s\frac{\tilde{g}_n(X_i)-g(X_i)}{g(X_i)}\right)^{-2}ds . 
\]
Since convergence of $\hat{g}_n$ to $ g$ is valid pointwise and in $\mathbb{L}_\infty$ norm (see Remark~\ref{rem:first}), and since $\tilde g_n$ is a slight modification of $\hat g_n$,  we have almost surely, for $n$
  large enough and for all $ s \in [0,1]$ and all $x\in [0,1]$, 
\begin{align*}
1 + s\frac{\tilde{g}_n(x)-g(x)}{g(x)} \geq 1 -
s\frac{\|\hat{g}_n-g\|_{\infty}}{\theta} \geq 1 - \frac{s}{2} > 0.
\end{align*}
Hence, for all $x\in [0,1]$ and  large enough $n$, 
\begin{equation*}
\int_0^1\left(1 + s\frac{\tilde{g}_n(x)-g(x)}{g(x)}\right)^{-2}ds \leq
\int_0^1\frac{4ds}{(2-s)^2} = 2,
\end{equation*}
and we obtain 
\begin{equation}
\label{eq:bound-gcp-g}
\Big|\frac{1}{\tilde{g}_n(X_i)}-\frac{1}{g(X_i)}\Big| 
\leq 2 \delta^{-2}|\tilde{g}_n(X_i)-g(X_i)|\quad \mbox{a.s.}
\end{equation}
%%%%%%%%%%%%%%%
%%%%%%%%%%%%%%%
We also  use the following lemma, whose proof is postponed to Appendix~\ref{sec:appen_sum_hat_tau}. 
\begin{lem}\label{lem:sum_hat_tau}
% By denoting $\|\hat{\tau}-\tau\|_{\infty,[0,1]} =\underset{x \in  [0,1]}{\sup}|\hat{\tau}(x)-\tau(x)|$, w
For large enough $n$, we have  
\begin{equation}
%\label{eq:toto18}
%\|\hat{\tau}-\tau\|_{\infty,[0,1]} &\leq  c_6\left(|\hat{\theta} - \theta|
%  +\|\hat{g}_n-g\|_{\infty}\right)\quad \mbox{a.s.}\\
\label{eq:bound-sum-taucp}
\frac{n}{|\sum_k\hat{\tau}_k|} \leq c_7 \quad \mbox{a.s.}
\end{equation}
\end{lem}
%%%%%%%%%
%%%%%%%%%%
By returning to Equality~\eqref{eq:equality-difference} and combining with \eqref{eq:bound-tau-g}, \eqref{eq:bound-sum-tau}, \eqref{eq:bound-gcp-g} and \eqref{eq:bound-sum-taucp}, we obtain  
\begin{eqnarray}
\label{eq:inequality-difference}
|\hat{f}_n(x) - f_2(x)|^2 &\leq& c_8\left(\frac{1}{nh}\sum_{i=1}^n |\tilde{g}_n(X_i)-g(X_i)| \times \Big| K\Big(\frac{x-X_i}{h}\Big) \Big|\right)^2\notag\\
&&+c_9|\hat{\theta}_n -
\theta|^2\left(\frac{1}{nh}\sum_{i=1}^n\Big|K\Big(\frac{x-X_i}{h}\Big)\Big|\right)^2\\
&&+c_{10}\left(\frac{1}{n}\sum_{i=1}^n
  |\tilde{g}_n(X_i)-g(X_i)|\right)^2 \left(\frac{1}{nh}\sum_{i=1}^n K\Big(\frac{x-X_i}{h}\Big)\right)^2\quad \mbox{a.s.} \notag
\end{eqnarray}
We now successively control the expectations  $T_1,T_2$ and $T_3$ of the three terms appearing in this upper-bound. For the first term, we have
\begin{eqnarray*}
T_1&=&\esp_{\theta,f}\left[\left(\frac{1}{nh}\sum_{i=1}^n
    |\tilde{g}_n(X_i)-g(X_i)| \times \Big| K\Big(\frac{x-X_i}{h}\Big) \Big| \right)^2\right]\\
&=& \esp_{\theta,f}\left[\frac{1}{n^2h^2}\sum_{i,j=1}^n |\tilde{g}_n(X_i)-g(X_i)| |\tilde{g}_n(X_j)-g(X_j)| \times \Big| K\Big(\frac{x-X_i}{h}\Big)  K\Big(\frac{x-X_j}{h}\Big)\Big|\right]\\
&=& \frac{1}{nh}\esp_{\theta,f}\left[\frac{1}{h} |\tilde{g}_n(X_1)-g(X_1)|^2
  K^2\Big(\frac{x-X_1}{h}\Big)\right] \\
&&+\frac{n-1}{n}\esp_{\theta,f}\left[\frac{1}{h^2}|\tilde{g}_n(X_1)-g(X_1)| |\tilde{g}_n(X_2)-g(X_2)|\times \Big| K\Big(\frac{x-X_1}{h}\Big)  K\Big(\frac{x-X_2}{h}\Big) \Big| \right]. 
\end{eqnarray*}
Now, 
\begin{eqnarray}
\label{eq:expectation-type1}
T_{11}&=&\esp_{\theta,f}\left[ \frac{1}{h}|\tilde{g}_n(X_1)-g(X_1)|^2
  K^2\Big(\frac{x-X_1}{h}\Big)\right] \notag\\
&=&\int_0^1
\esp_{\theta,f}\left( |\hat{g}_{n-1}(t)-g(t)|^2\right)
K^2\Big(\frac{x-t}{h}\Big)\frac{g(t)}{h}dt \ \text{ (according to definition~\eqref{eq:def_weight})} \notag\\
&\leq& C_{10}n^{\frac{-2\beta}{2\beta+1}} \int_0^1 K^2\Big(\frac{x-t}{h}\Big)\frac{g(t)}{h}dt \ \text{ (according to Remark~\ref{rem:first})}
\notag\\
&\leq& C_{11}n^{\frac{-2\beta}{2\beta+1}}\ \text{(according to Lemma \ref{lem:Bochner})},
\end{eqnarray}
and in the same way 
\begin{align*}
T_{12}=&\esp_{\theta,f}\left[\frac{1}{h^2}|\tilde{g}_n(X_1)-g(X_1)|
  |\tilde{g}_n(X_2)-g(X_2)| K\Big(\frac{x-X_1}{h}\Big)
  K\Big(\frac{x-X_2}{h}\Big)\right] \\
= & \int_0^1\int_0^1\esp_{\theta,f}\Big[\Big|\frac{n-2}{n-1}\hat{g}_{n-2}(t)-g(t)+\frac{1}{(n-1)h}K\Big(\frac{t-s}{h}\Big)\Big| \\
& \times \Big|\frac{n-2}{n-1}\hat{g}_{n-2}(s)-g(s)
  +\frac{1}{(n-1)h}K\Big(\frac{s-t}{h}\Big)\Big|\Big]\Big| K\Big(\frac{x-t}{h}\Big)   K\Big(\frac{x-s}{h}\Big) \Big|\frac{g(t)g(s)}{h^2}dtds .
\end{align*}
This last term is upper-bound by  
\begin{align*}
T_{12}\leq & \int_0^1\int_0^1\esp_{\theta,f}\Big[\Big(|\hat{g}_{n-2}(t)-g(t)|+\frac{1}{n-1}g(t)+\frac{1}{(n-1)h}\Big|K\Big(\frac{t-s}{h}\Big)\Big|\Big)\\
& \times \Big(|\hat{g}_{n-2}(s)-g(s)|
 +\frac{1}{n-1}g(s)+\frac{1}{(n-1)h}\Big|K\Big(\frac{s-t}{h}\Big)\Big|\Big)\Big]\\
& \times  \Big| K\Big(\frac{x-t}{h}\Big) K\Big(\frac{x-s}{h}\Big) \Big| \frac{g(t)g(s)}{h^2}dtds\\
\leq & \int_0^1\int_0^1\left\{\esp_{\theta,f}^{1/2}\big[|\hat{g}_{n-2}(t)-g(t)|^2\big]\esp_{\theta,f}^{1/2}\big[|\hat{g}_{n-2}(s)-g(s)|^2\big]+o\Big(\frac{1}{nh}\Big) \right\}\\
& \times \Big| K\Big(\frac{x-t}{h}\Big)K\Big(\frac{x-s}{h}\Big) \Big| \frac{g(t)g(s)}{h^2}dtds.
\end{align*}
According to Remark~\ref{rem:first}, we have
\begin{equation}
\label{eq:expectation-type2}
T_{12}\leq C_{12}n^{\frac{-2\beta}{2\beta+1}} \left[\int_0^1 \Big| K\Big(\frac{x-t}{h}\Big) \Big|\frac{g(t)}{h}dt\right]^2 \leq C_{13}n^{\frac{-2\beta}{2\beta+1}}\ \text{(according to Lemma \ref{lem:Bochner})}.
\end{equation}
Thus we get that
\begin{equation}
\label{eq:term1}
T_1= \esp_{\theta,f}\left[\left(\frac{1}{nh}\sum_{i=1}^n
    |\tilde{g}_n(X_i)-g(X_i)| \Big| K\Big(\frac{x-X_i}{h}\Big)\Big| \right)^2\right] \leq C_{14}n^{\frac{-2\beta}{2\beta+1}}.
\end{equation}
For the second term in the right hand side of~\eqref{eq:inequality-difference}, we have
\begin{eqnarray*}
T_2 &=& \esp_{\theta,f}\left[|\hat{\theta}_n -
\theta|^2\left(\frac{1}{nh}\sum_{i=1}^n\Big|K\Big(\frac{x-X_i}{h}\Big)\Big|\right)^2\right]\\
 &\leq& \esp_{\theta,f}^{1/2}\left[|\hat{\theta}_n -
\theta|^4\right]\esp_{\theta,f}^{1/2}\left[\left(\frac{1}{nh}\sum_{i=1}^n\Big|K\Big(\frac{x-X_i}{h}\Big)\Big|\right)^4\right] .
\end{eqnarray*}
The proof of the following lemma is postponed to Appendix~\ref{sec:appen_control_K1}. 
\begin{lem}\label{lem:control_K1}
There exist some positive constant $C_{15}$ such that   
\begin{equation}
\label{eq:toto}
\esp_{\theta,f}\left[\left(\frac{1}{nh}\sum_{i=1}^n\Big|K\Big(\frac{x-X_i}{h}\Big)\Big|\right)^4\right]
\leq C_{15}.
\end{equation}
\end{lem}
%%%%%
This lemma entails that
\begin{equation}
\label{eq:term2}
T_2
%= \esp_{\theta,f}\left[|\hat{\theta} -\theta|^2\left(\frac{1}{nh}\sum_{i=1}^n\Big|K\Big(\frac{x-X_i}{h}\Big)\Big|\right)^2\right]
\leq C_{15}\left[\esp_{\theta,f}\left(|\hat{\theta}_n -    \theta|^4\right)\right]^{\frac{1}{2}}. 
\end{equation}
%%%%%%%%%%%%%
%%%%%%%
Now, we turn to the third term in the right hand side of~\eqref{eq:inequality-difference}. We have
\begin{eqnarray*}
T_3 &=&\esp_{\theta,f}\left[\left(\frac{1}{n}\sum_{i=1}^n
  |\tilde{g}_n(X_i)-g(X_i)|\right)^2 \left(\frac{1}{nh}\sum_{i=1}^n
  \Big| K\Big(\frac{x-X_i}{h}\Big) \Big| \right)^2\right]\\
&=& \esp_{\theta,f}\left[\frac{1}{n^4h^2}\sum_{i,j,k,l=1}^n |\tilde{g}_n(X_i)-g(X_i)| |\tilde{g}_n(X_j)-g(X_j)| \Big|  K\Big(\frac{x-X_k}{h}\Big)  K\Big(\frac{x-X_l}{h}\Big) \Big|  \right] .
\end{eqnarray*}
By using the same arguments as for obtaining~\eqref{eq:expectation-type1} and~\eqref{eq:expectation-type2}, we can get that
\begin{equation}
\label{eq:term3}
T_3 \leq C_{16} n^{\frac{-2\beta}{2\beta+1}}.
\end{equation}
According to \eqref{eq:term1}, \eqref{eq:term2} and \eqref{eq:term3}, we may conclude 
\begin{equation}
\label{eq:difference}
\esp_{\theta,f}(|\hat{f}_n(x)-f_2(x)|^{2}) \leq C_{15}\Big[\esp_{\theta,f}\Big(|\hat{\theta}_n -
    \theta|^4\Big)\Big]^{\frac{1}{2}} + C_{17}n^{\frac{-2\beta}{2\beta+1}}. 
\end{equation}
By returning to Inequality~\eqref{eq:toto20} and combining it with \eqref{eq:toto21} and \eqref{eq:difference}, we achieve that
\begin{equation*}
\esp_{\theta,f}(|\hat{f}_n(x)-f(x)|^{2}) \leq C_1\Big[\esp_{\theta,f}\Big(|\hat{\theta}_n -
    \theta|^4\Big)\Big]^{\frac{1}{2}} + C_2n^{\frac{-2\beta}{2\beta+1}}. 
\end{equation*}
\end{proof}

%%%%%%%
\subsection{Other proofs} \label{sec:other_proofs}

\begin{proof}[Proof of Proposition~\ref{prop:hatft_decreases}]
By using the same arguments as for obtaining~\eqref{eq-decroissante}, we can get that
\begin{equation*}
l_n(\hat{f}^{(t)}) - l_n(\hat{f}^{(t+1)}) \geq \frac{1}{n}\sum_{k=1}^n \hat{\omega}_k^{(t)} D(\hat{f}^{(t+1)}\mid \hat{f}^{(t)}).
\end{equation*}
Let us now denote by 
 \begin{equation*}
m = \underset{x \in [-1,1]}{\inf} K_h(x)\ \text{and}\ M = \underset{x \in [-1,1]}{\sup} K_h(x),
\end{equation*} 
then $m$ and $M$ are two positive constants depending on the bandwidth $h$ and the
kernel $K$. We note that for all $x \in [0,1]$,
\begin{equation*}
m \leq \int_0^1 K_h(u-x)du \leq \min (M,1).
\end{equation*}  
Thus, for all $t\ge 1$, the estimate $\hat{f}^{(t)}$ is lower bounded by $m$. Since the operator $\mathcal{N}$ is increasing,  it follows that $\mathcal{N}\hat{f}^{(t)}$ is also lower bounded by $m$.
% \begin{eqnarray*}
% \mathcal{N}f^{t}(u) &=& \exp\Big\{\frac{\int_0^1K_h(x-u)\log
%     f^t(x)dx}{\int_0^1K_h(s-u)ds}\Big\}\\
% &\geq&\exp\Big\{\frac{\int_0^1K_h(x-u)\log m dx}{\int_0^1K_h(s-u)ds}\Big\}\\
% &\geq& m.
% \end{eqnarray*}
% Let us denote by
% \begin{equation*}
% c = \frac{(1-\theta) m}{\theta +
%   (1-\theta) m}.
% \end{equation*}
Now the function 
\begin{equation*}
x \mapsto \frac{(1-\theta)x}{\theta+(1-\theta)x}
\end{equation*} 
is increasing, so that we finally obtain
\begin{equation*}
\hat{\omega}_k^{(t)} = \frac{(1-\theta)\mathcal{N}\hat{f}^{(t)}(X_k)}{\theta +
  (1-\theta)\mathcal{N}\hat{f}^{(t)}(X_k)}
\geq \frac{(1-\theta) m}{\theta +
  (1-\theta) m} = c.
\end{equation*}
This concludes the proof. 
\end{proof}

%%%%%%%
%%%%%%%%%%%%

\begin{proof}[Proof of Proposition~\ref{prop:subsequence}.]
We start by stating a  lemma,  whose proof  is postponed  to
Appendix~\ref{sec:appen_continuous}.
\begin{lem}
\label{lem-continuous}
%\red{pas besoin de $G$ continu a priori}
The function $l: \mathcal{B} \rightarrow \mathbb{R}$ 
%and the operator $G: \mathcal{B} \rightarrow \mathcal{B}$ are 
is continuous
with respect to the topology induced by uniform convergence on the set of 
functions defined on $[0,1]$.
\end{lem}
First, for all $f \in \mathcal{B}$, we remark that $m\leq f(\cdot)\leq M/m$. Thus, $\mathcal{N}(f)$ and $l(f)$ are well-defined for $f \in \mathcal{B}$. Moreover, it is easy to see that $l(f)$ is bounded below on $\mathcal{B}$. According to the definition~\eqref{eq:f^t} of the sequence $\{f^t\}_{t\ge 0}$, every
function $f^t$ belongs to $\mathcal{B}$. As a consequence, we obtain that the sequence $\{l(f^t)\}_{t\ge 0}$  is decreasing and lower bounded, thus it is convergent and the sequence  $\{f^t\}_{t\ge 0}$ converges (simply) to a local minimum of $l$.   

Now, it is easy to see that $l$ is a strictly convex function on the convex set $\mathcal{B}$ (relying on \cite{Eggermont99}). Existence and uniqueness of the minimum $f^\star$ of $l$ in $\mathcal{B}$ thus follows, as well as the simple convergence of the iterative sequence $\{f^t\}_{t\ge 0}$ to this unique minimum.

For all $x, y \in [0,1]$ and for all
$t$, we have
\begin{eqnarray*}
| f^t(x) -f^t(y) | &=& \frac{1}{\int_0^1
  \omega_t(u)g_0(u)du}\Big | \int_0^1
\frac{[K_h(u-x)-K_h(u-y)]\omega_t(u)g_0(u)}{\int_0^1K_h(s-u)ds}du
\Big |\\
&\leq& \frac{1}{\int_0^1
  \omega_t(u)g_0(u)du} \int_0^1
\frac{| K_h(u-x)-K_h(u-y)| \omega_t(u)g_0(u)}{m}du\\
%&\leq& \frac{1}{\int_0^1 \omega_t(u)g_0(u)du} \int_0^1
%\frac{L |K_h(u-x)-K_h(u-y)|}{m}\omega_t(u)g_0(u)du\\
&\leq& \frac{L}{m}| x-y |,
\end{eqnarray*}
so that the sequence $\{f^t\}$ is uniformly bounded and
equicontinuous. Relying on Arzel\`{a}-Ascoli theorem, there exists a
subsequence  $\{f^{t_k}\}$  of  $\{f^t\}$  which  converges  uniformly to some limit. However, this uniform limit must be the simple limit of the sequence, namely the minimum $f^\star$ of $l$. Now, uniqueness of the uniform limit value of the sequence $\{f^t\}_{t\ge 0}$ entails its convergence.  
% %%
% Now we write,
% % Furthermore,
% \begin{eqnarray*}
% l(f^{t_k}) - l(f^*) &=&l(G(f^*)) - l(f^*)+ l(f^{t_k}) - l(G(f^*))\\
% &\leq&-c D(G(f^*)|f^*)+ l(f^{t_k}) - l(G(f^{t_{k-1}})) + l(G(f^{t_{k-1}})) -
% l(G(f^*))\\
% &\leq&-c D(G(f^*)|f^*)+ l(f^{t_k}) - l(f^{t_{k-1}+1}) + l\circ
% G(f^{t_{k-1}}) - l\circ G(f^*)\\
% &\leq & -c D(G(f^*)|f^*)+ l\circ
% G(f^{t_{k-1}}) - l\circ G(f^*) ,
% \end{eqnarray*}
% where the last bound comes from the fact that 
% the sequence $\{l(f^t)\}_{t\ge 0}$ is decreasing and $t_k\ge t_{k-1}+1$. 
% Since both $l$ and $l\circ G$ are continuous, we know that $l(f^{t_k})  - l(f^*)$ and  $l\circ G(f^{t_{k-1}}) - l\circ G(f^*) $ both converge to 0.
% This entails that $ D(G(f^*)|f^*) = 0$. This means that
% $f^*$ is a fixed point of $G$.

\end{proof}

%%%%%%
%%%%%%

\appendix                    

\section{Proofs of technical lemmas}\label{sec:appendix}

\subsection{Proof of Lemma \ref{lem-expectation-Sn}}\label{sec:appen_Sn}
\begin{proof}
We first show~\eqref{eq:toto10}.  
%that for $n$ large enough
%\begin{equation*}
%\esp_{\theta,f}\big[\frac{1}{S_n^2}\big] \leq \frac{c_2}{n^2}.
%\end{equation*}
According to the law of large numbers, since
$\esp_{\theta,f}\big(f(X_1)/g(X_1)\big) = 1$, we have
\begin{align}
\label{eq:toto27}
\frac {S_n} {n} = \frac {1} {n}
\displaystyle\sum_{i=1}^n \frac {f(X_i)} {g(X_i)} \xrightarrow[n\to\infty] {as} 1.
\end{align}
Hence 
\begin{equation*}
\frac{n^2}{S_n^2} = \big(\frac{S_n}{n}\big)^{-2} \xrightarrow[n\to\infty] {as} 1.
\end{equation*}
By the dominated convergence theorem, there exists a constant $c_2 >
0$ such that for $n$ large enough
\begin{equation*}
\esp_{\theta,f}\big[\frac{1}{S_n^2}\big] =\frac{1}{n^2} \esp_{\theta,f}\big[\frac{n^2}{S_n^2}\big]
\leq \frac{c_2}{n^2},
\end{equation*}
establishing~\eqref{eq:toto10}. 
Let us now prove~\eqref{eq:toto9}.
%show that for $n$ large enough
%\begin{equation*}
%\esp_{\theta,f}\big[\frac{1}{S_n}\big] \leq \frac{1}{n} + \frac{c_1}{n^2}.
%\end{equation*}
By using a  Taylor's expansion, we have
\begin{eqnarray*}
\frac {1} {S_n} = \frac {1} {n} \times \frac {1} {1 + (\frac {S_n} {n} - 1)} = \frac {1} {n}\left[2 -  \frac {S_n} {n} + \Big(\frac {S_n} {n} - 1\Big)^2 \frac
{1}{(1 + \gamma_n(\frac{S_n} {n} - 1))^3}\right] ,
\end{eqnarray*}
where $\gamma_n \in ]0, 1[$ depends on $S_n$. Combining this with
\eqref{eq:toto27}, we obtain
\begin{equation*}
\frac {1}{(1 + \gamma_n(\frac{S_n}{n} - 1))^3} \xrightarrow[n\to\infty] {as} 1.
\end{equation*}
Thus, there exist some positive constants $c, c'$ such that for $n$
large enough,
\begin{eqnarray}
\label{eq:toto28}
\frac {1} {n}\Big[2 -\frac {S_n} {n} + c'\big(\frac {S_n} {n} - 1\big)^2\Big] \leq \frac {1} {S_n} \leq \frac {1} {n}\Big[2 -\frac {S_n} {n} + c\big(\frac {S_n} {n} - 1\big)^2\Big]\quad\text{a.s.}
\end{eqnarray}
This implies in particular that 
\begin{eqnarray*}
\esp_{\theta,f}\Big[\frac {1} {S_n}\Big] \leq \frac {1} {n}\Big[2 -  \frac {\esp_{\theta,f}[S_n]} {n} +
c\esp_{\theta,f}\big[(\frac {S_n} {n} - 1)^2\big]\Big] = \frac {1} {n} + \frac{c}{n}\esp_{\theta,f}\Big[(\frac {S_n} {n} - 1)^2\Big].
\end{eqnarray*}
In addition,
\begin{eqnarray*}
\esp_{\theta,f}\Big[(\frac {S_n} {n} - 1)^2\Big] = \Var\Big(\frac{S_n}{n}\Big) = \frac{1}{n}\Var\left(\frac{f(X_1)}{g(X_1)}\right).
\end{eqnarray*}
Remember that the ratio $f/g$ is bounded (by $\delta^{-1}$) and thus has finite variance. Hence, there exists a positive constant $c_1$ such that for $n$ large enough
\begin{equation*}
\esp_{\theta,f}\big[\frac{1}{S_n}\big] \leq \frac{1}{n} + \frac{c_1}{n^2}.
\end{equation*}
We now prove~\eqref{eq:toto11}. 
%that for $n$ large enough
%\begin{equation*}
%\esp_{\theta,f}\Big[\frac {1} {S_n+b}\Big] \geq \frac{1}{n} - \frac{C_2}{n^2}.
%\end{equation*}
By using again a Taylor expansion, we have
\begin{equation*}
\frac {1} {S_n + \delta^{-1}} 
= \frac {1} {S_n} \times \frac {1} {1 + 1/(\delta S_n)}
= \frac {1} {S_n} - \frac {1} {\delta S_n^2} \times \frac {1} {[1 + \beta_n / (\delta S_n)]^2},
\end{equation*}
where $\beta_n \in ]0,1[$ depends on $S_n$. We also have 
\begin{equation*}
\frac {1} {[1 + \beta_n / (\delta S_n)]^2}
%\frac {1}{(1 + \beta_n\frac{b} {S_n})^2} 
\xrightarrow[n\to\infty] {as} 1.
\end{equation*}
Thus, there exists a positive constant $c''$ such that for $n$ large enough
\begin{eqnarray*}\
\esp_{\theta,f}\Big[\frac {1} {S_n + \delta^{-1}} \Big]
= \esp_{\theta,f}\Big[\frac {1} {S_n} - \frac {1} {\delta S_n^2} \times \frac {1} {[1 + \beta_n /(\delta S_n)]^2} \Big]
\geq \esp_{\theta,f}\Big[\frac {1} {S_n}\Big] - \esp_{\theta,f}\Big[\frac {c''}{S_n^2}\Big] \quad\text{a.s.}
\end{eqnarray*}
According to \eqref{eq:toto28}, we have
\begin{eqnarray*}
\esp_{\theta,f}\Big[\frac {1} {S_n}\Big] \geq \frac {1} {n}\left[2 -  \frac {\esp_{\theta,f}[S_n]} {n} +
c' \esp_{\theta,f}\Big[(\frac {S_n} {n} - 1)^2\Big]\right] = \frac {1} {n} + \frac{c'}{n^2}\Var\left(\frac{f(X_1)}{g(X_1)}\right),
\end{eqnarray*}
and it is proved above that 
\begin{equation*}
\esp_{\theta,f}\Big[\frac{1}{S_n^2}\Big] \leq \frac{c_2}{n^2}. %\quad \text{where $c_4$ is a positive constant}.
\end{equation*}
Thus we obtain Inequality~\eqref{eq:toto11}, namely 
\begin{equation*}
\esp_{\theta,f}\Big[\frac {1} {S_n+\delta^{-1}}\Big] \geq \frac{1}{n} - \frac{c_3}{n^2}.
\end{equation*}
Finally, we show~\eqref{eq:toto12}. 
% that for $n$ large enough
% \begin{equation*}
% \esp_{\theta,f}\big[\frac {1} {S_n^2}\big] - \esp_{\theta,f}^2\big[\frac {1} {b + S_n}\big] =
%\end{equation*}
% We have proved above that for $n$ large enough 
In the same way as we proved~\eqref{eq:toto11} above, we have for large enough $n$, 
 \begin{equation*}
 \esp_{\theta,f}\Big[\frac {1} {S_n+2\delta^{-1}}\Big] \geq \frac{1}{n} - \frac{c_3'}{n^2} > 0
 \end{equation*}
and thus
\begin{equation}
\label{eq:toto29}
\esp_{\theta,f}^2\Big[\frac {1} {S_n+2 \delta^{-1}}\Big] \geq \frac{1}{n^2}\left(1 - \frac{2c_3'}{n} +
\frac{c_3'^2}{n^2}\right) \geq \frac{1}{n^2}\Big(1 - \frac{2c_3'}{n}\Big).
\end{equation}
According to Inequality~\eqref{eq:toto28} (containing only positive terms for $n$ large enough), we have
\begin{eqnarray*}
\frac {1} {S_n^2} &\leq& \frac {1} {n^2}\left[4 +\frac {S_n^2} {n^2} +
c^2\Big(\frac {S_n} {n} - 1\Big)^4 -4\frac{S_n}{n} +4c\Big(\frac {S_n} {n} -
1\Big)^2 - 2c\frac{S_n}{n}\Big(\frac {S_n} {n} - 1\Big)^2\right]\quad\text{(as)}\\
&\leq& \frac {1} {n^2}\left[4 +\frac {S_n^2} {n^2} +
c^2\Big(\frac {S_n} {n} - 1\Big)^4 -4\frac{S_n}{n} +4c\Big(\frac {S_n} {n} -
1\Big)^2\right]\quad\text{a.s.}
\end{eqnarray*}
Since 
\begin{align*}
\esp_{\theta,f}[S_n] = n,\quad \esp_{\theta,f}[S_n^2] =
n\Var\left(\frac{f(X_1)}{g(X_1)}\right)+n^2\quad \text{and}\quad \esp_{\theta,f}\left[\Big(\frac {S_n} {n} - 1\Big)^2\right] = \frac{1}{n}\Var\left(\frac{f(X_1)}{g(X_1)}\right), 
\end{align*}
we have
\begin{eqnarray}
\label{eq:toto30}
\esp_{\theta,f}\Big[\frac {1} {S_n^2}\Big] &\leq& \frac {1} {n^2}\left[4 +\frac {\esp_{\theta,f}[S_n^2]} {n^2} +
c^2\esp_{\theta,f}\Big[\Big(\frac {S_n} {n} - 1\Big)^4\Big] -4\frac{\esp_{\theta,f}[S_n]}{n}
+4c\esp_{\theta,f}\Big[\Big(\frac {S_n} {n}- 1\Big)^2\Big]\right] \notag\\
&\leq&\frac {1} {n^2}\left[4
+\frac{1}{n}\Var\left(\frac{f(X_1)}{g(X_1)}\right) +1 +
c^2\esp_{\theta,f}\Big[\Big(\frac {S_n} {n} - 1\Big)^4\Big] - 4 +
\frac{4c}{n}\Var\left(\frac{f(X_1)}{g(X_1)}\right)\right]\notag \\
&\leq& \frac {1} {n^2}\left[1 +  \frac{C_4}{n} +
c^2\esp_{\theta,f}\Big[\Big(\frac {S_n} {n} - 1\Big)^4\Big]\right].
\end{eqnarray}
Combining \eqref{eq:toto29} and \eqref{eq:toto30}, we get that
\begin{equation}
\label{eq:toto31}
\esp_{\theta,f}\Big[\frac{1}{S_n^2}\Big] - \esp_{\theta,f}^2\Big[\frac {1} {S_n+2\delta^{-1}}\Big] \leq \frac{C}{n^3} + \frac{c^2}{n^2}\esp_{\theta,f}\Big[\Big(\frac {S_n} {n} - 1\Big)^4\Big].
\end{equation}
We now upper-bound  the quantity $\esp_{\theta,f}\big[(\frac {S_n} {n} - 1)^4\big]$.
Let us denote by
\begin{equation*}
U_i = \frac{f(X_i)}{g(X_i)} - 1.
\end{equation*}
We have
\begin{eqnarray*}
\left(\frac {S_n} {n} - 1\right)^4 &=& \frac{1}{n^4}\left(\displaystyle\sum_{i=1}^nU_i\right)^4 = \frac{1}{n^4}\displaystyle\sum_{i=1}^nU_i^4 +
\frac{1}{n^4}\displaystyle\sum_{i\neq j}^nU_i^3U_j+\\
& & + \frac{1}{n^4}\displaystyle\sum_{i\neq j}^nU_i^2U_j^2 +
\frac{1}{n^4}\displaystyle\sum_{i\neq j\neq k}^nU_i^2U_jU_k +  \frac{1}{n^4}\displaystyle\sum_{i\neq j\neq k\neq l}^nU_iU_jU_kU_l.
\end{eqnarray*}
Since the random variables $U_i$ are iid with mean
zero, we obtain
\begin{equation}
\label{eq:toto32}
\esp_{\theta,f}\left[\Big(\frac {S_n} {n} - 1\Big)^4\right] = \frac{1}{n^4}\left[n\esp_{\theta,f}(U_1^4)
+n(n-1)\esp_{\theta,f}(U_1^2U_2^2)\right] = O\Big(\frac{1}{n^2}\Big).
\end{equation}
Finally, according to \eqref{eq:toto31} and \eqref{eq:toto32} we have
\begin{equation*}
\esp_{\theta,f}\Big[\frac{1}{S_n^2}\Big] - \esp_{\theta,f}^2\Big[\frac {1} {S_n+2\delta^{-1}}\Big] = O\Big(\frac{1}{n^3}\Big).
\end{equation*}
\end{proof}

%%%%
%%%%
\subsection{Proof of Lemma~\ref{lem:sum_hat_tau}} \label{sec:appen_sum_hat_tau}
\begin{proof}
We write
\begin{equation*}
\frac{1}{\sum_k\hat{\tau}_k} = \frac{1}{\sum_k\tau_k +
  \sum_k(\hat{\tau}_k - \tau_k)}=
\frac{1}{\sum_k\tau_k}-\frac{\sum_k(\hat{\tau}_k
  -\tau_k)}{(\sum_k\tau_k)^2}\times \int_0^1\left(1+s\frac{\sum_k(\hat{\tau}_k
  -\tau_k)}{\sum_k\tau_k}\right)^{-2}ds.
\end{equation*}
Let us establish that  $\|\hat{\tau} - \tau\|_{\infty,[0,1]} =\sup_{x\in [0,1]} |\hat{\tau}(x) - \tau(x)|$ converges almost surely to zero. 
Indeed, 
\begin{equation*}
\hat{\tau}(x)-\tau(x) 
= (\theta - \hat{\theta}_n)\frac{1}{g(x)} + \hat{\theta}_n\left(\frac{1}{g(x)} -
  \frac{1}{\tilde{g}_n(x)}\right)
\end{equation*}
and using the same argument as for establishing~\eqref{eq:bound-gcp-g}, we get 
that for $n$ large enough 
and  for all $x \in [0, 1]$,
\begin{equation*}
|\hat{\tau}(x)-\tau(x)| \leq  \frac{|\hat{\theta}_n - \theta|}{\theta} + 2|\hat{\theta}_n|\frac{\|\hat{g}_n-g\|_{\infty}}{\theta^2} \leq \delta^{-1}|\hat{\theta}_n - \theta| +2\delta^{-2} \|\hat{g}_n-g\|_{\infty}.
\end{equation*}
By using consistency of $\hat \theta_n$ and Remark~\ref{rem:first}, we obtain that $\|\hat{\tau} - \tau\|_{\infty,[0,1]}$ converges almost surely to zero. 
Now, 
\begin{eqnarray*}
\forall s \in [0, 1], \quad 1+s\frac{\sum_k(\hat{\tau}_k
  -\tau_k)}{\sum_k\tau_k} &\geq& 1-s\frac{n\|\hat{\tau}_k
  -\tau_k\|_{\infty,[0, 1]}}{\sum_k\tau_k}\\
   &\geq& 1-s\frac{2\|\hat{\tau}_k
  -\tau_k\|_{\infty,[0, 1]}}{\theta} \geq 1-\frac{s}{2} > 0 \quad \mbox{a.s.}
\end{eqnarray*}
We obtain that
\begin{eqnarray*}
\frac{n}{|\sum_k\hat{\tau}_k|} &\leq& \frac{n}{\sum_k\tau_k}+\frac{n\sum_k|\hat{\tau}_k  -\tau_k|}{(\sum_k\tau_k)^2}\times \int_0^1\left(1+s\frac{\sum_k(\hat{\tau}_k
  -\tau_k)}{\sum_k\tau_k}\right)^{-2}ds \\
 &\leq& \frac{n}{\sum_k\tau_k}+\frac{n^2\|\hat{\tau}
   -\tau\|_{\infty,[0,1]}}{(\sum_k\tau_k)^2}\times
 \int_0^1\left(1-\frac{s}{2}\right)^{-2}ds \\
 &\leq& \frac{2}{1-\theta}+\frac{8\|\hat{\tau}
   -\tau\|_{\infty,[0,1]}}{(1-\theta)^2} \leq c_7 \quad \mbox{a.s.}
\end{eqnarray*}
\end{proof}

%%%%
%%%%
\subsection{Proof of Lemma \ref{lem:control_K1}}\label{sec:appen_control_K1}
\begin{proof}
In order to prove \eqref{eq:toto}, let us consider iid random variables $U_1,\dots,U_n$  defined as
\begin{equation*}
U_i= \Big|K\left(\frac{x-X_i}{h}\right)\Big|.
\end{equation*}
For all $1 \leq p \leq 4$, we have 
\begin{equation*}
\esp_{\theta,f}(U_i^p)=\int \Big|K^p\left(\frac{x-t}{h}\right)\Big|g(t)dt=h\int
\big|K^p(t)\big|g(x+th)dt\leq C_{15}h.
\end{equation*}
We then write
\begin{equation}
\label{eq:toto23}
\Big(\frac{1}{nh}\sum_{i=1}^n\Big|K\Big(\frac{x-X_i}{h}\Big)\Big|\Big)^4 =
\frac{1}{n^4h^4}\Big(\sum_iU_i\Big)^4,
\end{equation}
where
\begin{equation*}
\Big(\sum_iU_i\Big)^4 =\sum_iU_i^4+\sum_{i\neq j}U_i^3U_j+\sum_{i\neq j}U_i^2U_j^2+\sum_{i\neq j\neq
  k}U_i^2U_jU_k+\sum_{i\neq j\neq k\neq l}U_iU_jU_kU_l.
\end{equation*}
And for all choice of the bandwidth $h > 0$ such that $nh \rightarrow \infty$,
\begin{align}
\label{eq:toto24}
&\esp_{\theta,f}\Big[\Big(\sum_iU_i\Big)^4\Big] \notag \\ 
= & n\esp_{\theta,f}(U_1^4)+n(n-1)\esp_{\theta,f}(U_1^3U_2)+n(n-1)\esp_{\theta,f}(U_1^2U_2^2)+\notag\\
&  +n(n-1)(n-2)\esp_{\theta,f}(U_1^2U_2U_3)+ %\notag\\
 n(n-1)(n-2)(n-3)\esp_{\theta,f}(U_1U_2U_3U_4)\notag\\
= & n\esp_{\theta,f}(U_1^4)+n(n-1)\esp_{\theta,f}(U_1^3)\esp_{\theta,f}(U_1)+n(n-1)\esp_{\theta,f}^2(U_1^2)+\notag\\
&  +n(n-1)(n-2)\esp_{\theta,f}(U_1^2)\esp_{\theta,f}^2(U_1)+%\notag\\
 n(n-1)(n-2)(n-3)\esp_{\theta,f}^4(U_1)\notag\\
\leq  &  C_{15}n^4h^4.
\end{align}
According to~\eqref{eq:toto23} and~\eqref{eq:toto24} we obtain the result. 
\end{proof}

%%%
%%%%

\subsection{Proof of Lemma \ref{lem-continuous}}\label{sec:appen_continuous}
\begin{proof}
Let $f$ be a function in $\mathcal{B}$ and $\{f_n\}$ be a sequence of
densities on $[0,1]$ such that $\|f_n - f\|_{\infty}
\xrightarrow[n\to\infty] {} 0$. Let us recall that every $f\in \mathcal{B}$ satisfies the bounds $m\le f\le M/m$.
We have
\begin{eqnarray*}
\mid l(f_n) -l(f) \mid &=& \Big| \int_0^1 g_0(x)\log \frac{\theta + (1-\theta)\mathcal{N}f(x)}{\theta + (1-\theta)\mathcal{N}f_n(x)}dx\Big|\\
&\leq&\int_0^1 g_0(x)  \Big| \log \Big\{1+\frac{(1-\theta)[\mathcal{N}f_n(x)-\mathcal{N}f(x)]}{\theta + (1-\theta)\mathcal{N}f_n(x)} \Big\}\Big| dx,
\end{eqnarray*}
and
\begin{eqnarray*}
\mid \mathcal{N}f_n(x) -\mathcal{N}f(x) \mid &=& \mathcal{N}f(x)\Big|
\exp \frac{\int_0^1 K_h(u-x)[\log f_n(u)-\log
  f(u)]du}{\int_0^1K_h(s-x)ds} - 1\Big|\\
&\leq& \frac M m \Big|
\exp \frac{\int_0^1 K_h(u-x)[\log f_n(u)-\log
  f(u)]du}{\int_0^1K_h(s-x)ds} - 1\Big| .
\end{eqnarray*}
For $|x|<\epsilon$ small enough, we have $|\log (1+x)| \leq
2|x|$ and $|\exp (x) - 1| \leq 2|x|$. Combining with the fact that $f$
is bounded, we get that
\begin{eqnarray*}
\Big| \int_0^1 K_h(u-x)[\log f_n(u)-\log
  f(u)]du\mid &\leq&\int_0^1 K_h(u-x)\Big| \log \big\{1 + \frac{f_n(u)-
  f(u)}{f(u)}\big\} \Big| du\\
&\leq& 2\|f_n - f\|_{\infty}
\end{eqnarray*}
and thus
\begin{equation*}
\|\mathcal{N}f_n - \mathcal{N}f\|_{\infty} \leq \frac{4M}{m^2} \|f_n - f\|_{\infty}.
\end{equation*}
We finally obtain
\begin{equation*}
\mid l(f_n) -l(f) \mid \leq C\|f_n - f\|_{\infty},
\end{equation*}
where $C$ is a   constant depending on $h, K$ and
$\theta$. 
%Using an argument nearly identical to the above proof, we may show that there exists a constant $C_5$ depending on $h,K,\theta$ such that $$\|Gf_n - Gf\|_{\infty} \leq C_5 \|f_n - f\|_{\infty}.$$
\end{proof}

\bibliographystyle{chicago}
\bibliography{localFDR}

\end{document}